\documentclass[preprint,11pt]{elsarticle}
\usepackage[colorlinks]{hyperref}
\usepackage{amssymb,amsthm,latexsym}
\usepackage{graphicx}
\usepackage{mathrsfs}
\usepackage{amsfonts}
\usepackage{graphicx}
\usepackage{amssymb}
\usepackage{longtable}
\usepackage{amsmath}
\usepackage{color}
\usepackage{setspace}
\usepackage{caption}
\usepackage[figuresright]{rotating}
\usepackage[misc]{ifsym}
\usepackage{}
\usepackage{amssymb}
\usepackage{amsfonts}
\usepackage{bbm}
\usepackage{color}
\usepackage{mathrsfs}
\usepackage{makecell}
\usepackage{mathrsfs}
\usepackage{longtable}
\usepackage{amsmath}
\usepackage{arydshln}
\usepackage[figuresright]{rotating}
\usepackage{supertabular}
\usepackage{booktabs}
\usepackage{url}
\theoremstyle{plain}
\newtheorem{theorem}{Theorem}[section]
\newtheorem{lemma}[theorem]{Lemma}

\newtheorem{corollary}[theorem]{Corollary}

\newcommand{\C}{{\mathcal{C}}}

{\theoremstyle{definition}
\newtheorem{definition}[theorem]{Definition}
\newtheorem{example}[theorem]{Example}

\newtheorem{remark}[theorem]{Remark}

}
\newcommand{\tabincell}[2]{\begin{tabular}{@{}#1@{}}#2\end{tabular}}

\usepackage{paralist}
\journal{Designs, Codes and Cryptography}
\usepackage{geometry}
\begin{document}

\begin{frontmatter}

\title{A new class of self-orthogonal linear codes and their applications}

\vspace{1cm}

\author{Yaozong Zhang, Dabin Zheng and Xiaoqiang Wang}
\ead{zyzsdutms@163.com, dzheng@hubu.edu.cn, waxiqq@163.com}

\address{ Key Laboratory of Intelligent Sensing System and Security, Ministry of Education,
Hubei Key Laboratory of Applied Mathematics, Faculty of Mathematics and Statistics, Hubei University,  Wuhan, 430062, China}

\tnotetext[label1]
{The corresponding author is Dabin Zheng. This work was supported in part by the National Key Research Development Program of China (No. 2021YFA1000600) and the National Natural
Science Foundation of China under Grant Number 62272148, and in part by the Natural Science Foundation of Hubei Province of China under Grant Number 2023AFB847 and the Sunrise Program of Wuhan under the Grant Number 2023010201020419.}

\begin{abstract}
Self-orthogonal codes are a  subclass of linear codes that are contained within their dual codes. Since self-orthogonal codes are widely used in quantum codes, lattice theory and linear complementary dual (LCD) codes, they have received continuous attention and research. In this paper, we construct a class of self-orthogonal codes by using the defining-set approach, and determine their explicit weight distributions and the parameters of their duals. Some considered codes are optimal according to the tables of best codes known maintained at \cite{Grassl} and a class of almost maximum distance separable (AMDS) codes from their duals are obtained. As applications,  we obtain a class of new quantum codes, which are MDS or AMDS according to the quantum Singleton bound under certain conditions. Some examples show that  the constructed quantum codes have the better parameters than known ones maintained at \cite{Bierbrauer}.  Furthermore, a new class of LCD codes are given, which are almost optimal according to the sphere packing bound.
\end{abstract}


\end{frontmatter}

\section{Introduction}

Let $q=p^s$ be an odd prime power, $\mathbb{F}_q$ denote a finite field with $q$ elements and $\mathbb{F}_q^n$ denote the $n$-dimensional vector space over $\mathbb{F}_q$. For any vector $\mathbf{x}=(x_1,x_2,\ldots,x_n)\in\mathbb{F}_q^n$, the Hamming weight of $\mathbf{x}$ is defined as $wt(\mathbf{x})=|\{1\leq i\leq n|x_i\neq0\}|$, and the Hamming distance between two vectors $\mathbf{x}$ and $\mathbf{y}$ is defined as $d(\mathbf{x},\mathbf{y})=wt(\mathbf{x}-\mathbf{y})$. An $[n,k,d]_q$ linear code $\mathcal{C}$ over $\mathbb{F}_q$ is a $k$-dimensional subspace of $\mathbb{F}_{q}^n$, where $d=\min_{\mathbf{c}\in\mathcal{C},\mathbf{c}\neq\mathbf{0}}wt(\mathbf{c})$ denotes the minimum Hamming distance of $\mathcal{C}$. An $[n,k,d]_q$ linear code is referred to as a maximum distance separable (MDS for short) code if $d=n-k+1$, and an almost maximum distance separable (AMDS for short) code if $d=n-k$. An $[n,k,d]_q$ linear code is {\it length-optimal} if there is no $[n-1,k,d]_q$ code, and {\it dimension-optimal} if there is no $[n,k+1,d]_q$ code, and {\it distance-optimal} if there is no $[n,k,d+1]_q$ code. In addition, an $[n,k,d]_q$ linear code is called an {\it optimal code} if it is length-optimal or dimension-optimal or distance-optimal. The well-known MDS conjecture states that if there is a nontrivial $[n,k]_q$ MDS code over $\mathbb{F}_q$, then $n\leq q+1$ for any odd $q$ \cite{Huffman2003}. In other words, an $[n,k,n-k]_q$ AMDS code whose length $n>q+1$ is an optimal code for any odd $q$.

Let $A_{i}=|\{\mathbf{c}\in\mathcal{C}|wt(\mathbf{c})=i,0\leq i\leq n\}|$ denote the number of codewords in $\mathcal{C}$ with Hamming weight $i$. We refer to the sequence $(1,A_1,\ldots,A_n)$ as the {\it weight distribution} of $\mathcal{C}$ and the polynomial $A(x)=1+A_1x+\cdots+A_nx^n$ as the {\it weight enumerator} of $\mathcal{C}$. The weight distribution of linear codes not only determines their error correction capability, but also assesses the probability of error detection and correction~\cite{Klove2007}. The weight distributions of linear codes have been extensively studied in the literature \cite{Ding2020,Gao2023,Gao2022,Heng2017,Heng2020,Liu2019,Sun2024,Zheng2015,Zheng2021} and references therein. A linear code is referred to as a {\it $t$-weight code} if $|\{A_i|A_i\neq0,1\leq i\leq n\}|=t$. Linear codes with few weights have consistently been a prominent research topic in coding theory duo to their intriguing and valuable applications in constructing secret sharing schemes \cite{Li2017,Yang2017}, association schemes \cite{Luo2018,Yang2017} and strongly regular graphs \cite{Heng2017,Zheng2021}.

Let $\mathcal{C}$ be an $[n,k,d]_q$ linear code. The dual code of $\mathcal{C}$ is defined by
$$\mathcal{C}^\bot=\{\mathbf{a}\in\mathbb{F}_q^n: \langle\mathbf{a},\mathbf{c}\rangle=0~{\rm for~all}~\mathbf{c}\in\mathcal{C}\},$$
where $\langle\cdot,\cdot\rangle$ denotes the standard inner product. Clearly, the dimension of $\mathcal{C}^\bot$ is equal to $n-k$. A linear code $\mathcal{C}$ is {\it self-orthogonal} if $\mathcal{C}\subseteq\mathcal{C}^\bot$. Self-orthogonal codes have valuable applications in various fields, such as quantum codes \cite{Calderbank1997}, lattices theory \cite{Wan}, LCD codes \cite{Heng2023,Heng12024}, row-self-orthogonal matrices \cite{Massey1992,Massey1998} and so on. Consequently, self-orthogonal codes have received extensive attention and in-depth research.
A linear code $\mathcal{C}$ over $\mathbb{F}_q$ is {\it $p$-divisible} if all codewords of $\mathcal{C}$ have Hamming weights that are divisible by  $p$. In \cite{Huffman2003}, authors established a necessary and sufficient condition for cyclic codes being self-orthogonal and derived some results to determine whether a binary or ternary linear codes are self-orthogonal based on the divisibility of their weights.
The following lemma presents a sufficient condition for a class of divisible codes being self-orthogonal.

\begin{lemma}\cite{Li2023}
\label{le2.1}
Let $\mathcal{C}$ be an $[n,k,d]$ linear code over $\mathbb{F}_{p^s}$ and $\mathbf{1}\in\mathcal{C}$, where $\mathbf{1}$ is the all-$1$ vector of length $n$. If $\mathcal{C}$ is $p$-divisible, then $\mathcal{C}$ is self-orthogonal.
\end{lemma}

Recently, Heng et al. \cite{Heng2024} constructed a class of linear codes using the {\it augmentation technique} as follows:
\begin{equation*}
\begin{aligned}
\mathcal{C}_D=\left\{\mathbf{c}(a;c)=({\rm Tr}_1^{m}(ax)_{x\in D}+c\mathbf{1}:a\in\mathbb{F}_{q^m},c\in\mathbb{F}_q\right\},
\end{aligned}
\end{equation*}
where $D=\{x\in\mathbb{F}_{q^m}:{\rm Tr}_{1}^t(x^N)=0\}$ is the defining-set with some positive integer $N$, and $\mathbf{1}$ is the all-$1$ vector of length $|D|$.
They determined the weight distribution of $\mathcal{C}_D$ and obtained some self-orthogonal codes for some special $N$. Inspired by the work of Heng et al., we consider  a class of linear codes over $\mathbb{F}_{p^{s_2}}$
as follows:
\begin{equation}\label{eq1.4}
\begin{aligned}
\mathcal{C}_D=\left\{\mathbf{c}(a,b;c)=({\rm Tr}_{s_2}^s(ax+by)_{(x,y)\in D}+c\mathbf{1}:a,b\in\mathbb{F}_{p^s},c\in\mathbb{F}_{p^{s_2}}\right\},
\end{aligned}
\end{equation}
where the defining-set is given by
\begin{equation}\label{eq1.5}
\begin{aligned}
D=\left\{(x,y)\in(\mathbb{F}_{p^s},\mathbb{F}_{p^s}):{\rm Tr}_{s_1}^s(x^2+y^2)=0\right\},
\end{aligned}
\end{equation}
and $s_1, s_2$ are positive divisors of $s$. In this paper, by using Gauss sums to analyze the solutions of some equations over finite fields, we determine the weight distribution of the linear code $\mathcal{C}_D$ defined by Eq.(\ref{eq1.4}) in the condition $s_2\,|\,s_1$ or $s_1\,|\,s_2$ and obtain a class of self-orthogonal codes under certain conditions. Furthermore, we determine the parameters of dual code $\mathcal{C}_D^\bot$ of $\mathcal{C}_D$, and show that $\mathcal{C}_D^\bot$ is a class of AMDS codes if $s=s_2$. Some examples of the considered codes and their duals
have best parameters according to the tables maintained at at \cite{Grassl}. As applications, we obtain a new class of quantum codes, which are MDS or AMDS according to the quantum Singleton bound under certain conditions. Some examples show that the constructed quantum codes have better parameters than known ones maintained at \cite{Bierbrauer}. Furthermore, we obtain a new class of LCD codes, which are almost optimal according to the sphere packing bound.

The remainder of this paper is organized as follows. In Section 2, we introduce some basic concepts and important conclusions over finite fields. In Section 3, we give some auxiliary lemmas that are necessary for the subsequent sections. In Section 4, we determine the weight distributions of a class of self-orthogonal codes $\mathcal{C}_D$, and study their dual codes. In Section 5, we construct a class of quantum and LCD codes with new parameters by
applying the constructed self-orthogonal codes. Finally, Section 6 concludes the paper.

\section{Preliminaries}
Throughout this paper, we adopt the following notation unless otherwise stated:
\begin{enumerate}
  \item [$\bullet$] $q=p^s$ is an odd prime power for a positive integer $s$ and $s_1,s_2$ are two positive divisors of $s$.
  \item [$\bullet$] $\mathbb{F}_{p^s}$ is a finite field with $p^s$ elements and $\mathbb{F}_{p^s}^*=\mathbb{F}_{p^s}\backslash\{0\}$ is the multiplicative group of $\mathbb{F}_{p^s}$.
  \item [$\bullet$] ${\rm Tr}_{a}^b$ denotes the trace function from $\mathbb{F}_{p^b}$ to $\mathbb{F}_{p^a}$, where $\mathbb{F}_{p^a}$ is a subfield of $\mathbb{F}_{p^b}$.
  \item [$\bullet$] $\chi,\chi'$ and $\chi''$ denote the additive characters over $\mathbb{F}_{p^s},\mathbb{F}_{p^{s_1}}$ and $\mathbb{F}_{p^{s_2}}$, respectively.
  \item [$\bullet$] $\eta,\eta'$ and $\eta''$ denote the quadratic characters over $\mathbb{F}_{p^s},\mathbb{F}_{p^{s_1}}$ and $\mathbb{F}_{p^{s_2}}$, respectively.
  \item [$\bullet$] $\mathcal{G},\mathcal{G}'$ and $\mathcal{G}''$ denote the values of quadratic Gaussian sums over $\mathbb{F}_{p^s},\mathbb{F}_{p^{s_1}}$ and $\mathbb{F}_{p^{s_2}}$, respectively.
\end{enumerate}

Let $\mathbb{F}_{q}$ be a finite field with $q$ elements. For any $a\in\mathbb{F}_{q}$, a function $\chi_a$ from $\mathbb{F}_{q}$ into the complex numbers defined by
$$\chi_a(x)=\zeta_p^{{\rm Tr}_{1}^s(ax)}$$
is called an {\it additive character} of $\mathbb{F}_{q}$, where ${\rm Tr}_{1}^s( \cdot )$ is the trace function from $\mathbb{F}_q$ to $\mathbb{F}_p$ and $\zeta_p=e^{\frac{2\pi\sqrt{-1}}{p}}$ is a primitive $p$-th root of unity. The $\chi_1$ is called the canonical additive character of $\mathbb{F}_q$. It is easy to show that $\chi_a(x_1+x_2)=\chi_a(x_1)\chi_a(x_2)$ for $x_1,x_2\in\mathbb{F}_q$. Moreover,
$$\sum_{x\in\mathbb{F}_q^*}\chi_a(x)=\left\{
                                     \begin{array}{ll}
                                       q-1, & \hbox{if $a=0$;} \\
                                       -1, & \hbox{if $a\in\mathbb{F}_q^*$.}
                                     \end{array}
                                   \right.
$$

Let $\omega$ be a fixed primitive element of $\mathbb{F}_q$. For each $j=0,1,\ldots,q-1$, a function $\psi_j$ defined by
$$\psi_j(\omega^k)=e^{\frac{2\pi\sqrt{-1}jk}{q-1}}~{\rm for}~k=0,1,\ldots, q-2, $$
is called a {\it multiplicative character} of $\mathbb{F}_q^*$. For completeness of the definition, we stipulate that $\psi_j(0)=0$. For each $j$, it is obvious that $\psi_j(x_1x_2)=\psi_j(x_1)\psi_j(x_2)$ for $x_1, x_2\in\mathbb{F}_q^*$. In particular, for odd $q$ the function $\eta:=\psi_{\frac{q-1}{2}}(\cdot)$ is called the quadratic character of $\mathbb{F}_q^*$.  It is easy to see that
$$\eta(a)=\left\{
            \begin{array}{ll}
              1, & \hbox{if $a\in \mathbb{F}_q^*$ is a square}; \\
              -1, & \hbox{if $a\in \mathbb{F}_q^*$ is a non-square}.
            \end{array}
          \right.
$$
For the convenience of narration, we regulate that $\eta(0)=0$. Let $\mathbb{F}_{p^{s_1}}$ be a subfield of $\mathbb{F}_{p^{s_2}}$. By the definition of $\eta''$ and $\eta'$, for any $b\in\mathbb{F}_{p^{s_1}}$, it is well-known that $\eta''(b)=\eta'(b)$ if $\frac{s_2}{s_1}$ is odd and $\eta''(b)=1$ if $\frac{s_2}{s_1}$ is even. In the following sections of this paper, this conclusion will be repeatedly used.

Let $\psi$ be a multiplicative and $\chi$ be an additive character of $\mathbb{F}_q$. The {\it Gaussian sum} $\mathcal{G}(\psi,\chi)$ is defined by
$$\mathcal{G}(\psi,\chi)=\sum_{c\in\mathbb{F}_q^*}\psi(c)\chi(c).$$
In general, determining the exact value of the Gauss sum $\mathcal{G}(\psi,\chi)$ is quite challenging. However, if $\psi=\eta$, then the quadratic Gaussian sum $\mathcal{G}(\eta,\chi)$  can be explicitly calculated
as follows:
\begin{lemma}\cite{Lidl1997}
Let $\eta$ be a quadratic character and $\chi$ be an additive character over $\mathbb{F}_{p^s}$. The value of $\mathcal{G}(\eta,\chi)$ is given as follows:
\begin{equation*}
\begin{aligned}
\mathcal{G}(\eta, \chi)=(-1)^{s-1}(\sqrt{-1})^{(\frac{p-1}{2})^2s}\sqrt{p^s}=\left\{
                                                                    \begin{array}{ll}
                                                                      (-1)^{s-1}\sqrt{p^s}, & \hbox{if $p\equiv1~({\rm mod}~4)$;} \\
                                                                      (-1)^{s-1}(\sqrt{-1})^s\sqrt{p^s}, & \hbox{if $p\equiv3~({\rm mod}~4)$.}
                                                                    \end{array}
                                                                  \right.
\end{aligned}
\end{equation*}
\end{lemma}

For the convenience and simplicity of subsequent sections, we let  $\mathcal{G}=(-1)^{s-1}(\sqrt{-1})^{(\frac{p-1}{2})^2s}\sqrt{p^s}$,  $\mathcal{G}' =(-1)^{s_1-1}(\sqrt{-1})^{(\frac{p-1}{2})^2s_1}\sqrt{p^{s_1}}$ and
 $\mathcal{G}'' =(-1)^{s_2-1}(\sqrt{-1})^{(\frac{p-1}{2})^2s_2}\sqrt{p^{s_2}}$.

\begin{lemma}\cite{Lidl1997}\label{le2.3}
Let $\eta$ be a quadratic character and $\chi$ be an additive character over $\mathbb{F}_q$. Let $h(x)=h_2x^2+h_1x+h_0\in\mathbb{F}_q[x]$ with $h_2\neq 0$. Then
$$\sum_{c\in\mathbb{F}_q}\chi(h(c))=\chi\left(h_0-h_1^2(4h_2)^{-1}\right)\eta(h_2)\mathcal{G}(\eta,\chi).$$
\end{lemma}

The following is a well-known bound for linear codes.
\begin{lemma}\cite[The sphere packing bound]{Huffman2003}\label{The sphere packing bound}
Let $\mathcal{C}$ be an $[n,k,d]_q$ linear code. Then
$$q^n\geq q^k\sum_{i=0}^{\lfloor\frac{d-1}{2}\rfloor}{n\choose i}(q-1)^i,$$
where $\lfloor\frac{d-1}{2}\rfloor$ denotes the largest integer less than or equal to $\frac{d-1}{2}$.
\end{lemma}

At the end of this section, we present the Pless power moments for linear codes, which are useful for determining the parameters of the duals of the proposed linear codes. Let $\mathcal{C}$ be an $[n,k,d]_q$ linear code, and denote its dual by $\mathcal{C}^\bot$. Let $(1,A_1,\ldots,A_n)$ and $(1,A_1^\bot,\ldots,A_n^\bot)$ represent the weight distributions of $\mathcal{C}$ and $\mathcal{C}^\bot$, respectively. The first four Pless power moments are given as follows \cite[Page~259]{Huffman2003}:
\begin{equation*}
\begin{aligned}
&\sum_{j=0}^nA_j=q^k;\\
&\sum_{j=0}^njA_j=q^{k-1}(qn-n-A_1^\bot);\\
&\sum_{j=0}^nj^2A_j=q^{k-2}[(q-1)n(qn-n+1)]-(2qn-q-2n+2)A_1^\bot+2A_2^\bot];\\
&\sum_{j=0}^nj^3A_j=q^{k-3}[(q-1)n(q^2n^2-2qn^2+3qn-q+n^2-3n+2)-(3q^2n^2-3q^2n-6qn^2\\
&~~~~~~~~~~~~~~+12qn+q^2-6q+3n^2-9n+6)A_1^\bot+6(qn-q-n+2)A_2^\bot-6A_3^\bot].\\
\end{aligned}
\end{equation*}
\section{Some auxiliary lemmas}
In this section, we present some auxiliary lemmas to determine the weight distribution of the linear code $\mathcal{C}_D$ defined by Eq.(\ref{eq1.4}).
\begin{lemma}\label{le3.1}
Let $\mathcal{C}_D$ be the linear code defined by Eq.(\ref{eq1.4}). Then the code length $n$ of $\mathcal{C}_D$ is equal to $\frac{p^{2s}+(p^{s_1}-1)\mathcal{G}^2}{p^{s_1}}$.
\end{lemma}
\begin{proof}
By the orthogonal property of additive characters and Lemma \ref{le2.3},
\begin{equation*}
\begin{aligned}
n=|D|&=|\left\{(x,y)\in(\mathbb{F}_{p^s},\mathbb{F}_{p^s}):{\rm Tr}_{s_1}^s(x^{2}+y^{2})=0\right\}|\\
     &=\frac{1}{p^{s_1}}\sum_{x,y\in\mathbb{F}_{p^s}}\sum_{z\in\mathbb{F}_{p^{s_1}}}\chi_{z}'({\rm Tr}_{s_1}^s(x^{2}+y^{2}))\\
     &=p^{2s-s_1}+\frac{1}{p^{s_1}}\sum_{z\in\mathbb{F}_{p^{s_1}}^*}\sum_{x,y\in\mathbb{F}_{p^s}}\chi_z(x^2+y^2)\\
&=p^{2s-s_1}+\frac{\mathcal{G}^2}{p^{s_1}}\sum_{z\in\mathbb{F}_{p^{s_1}}^*}\eta^{2}(z)\\
&=\frac{p^{2s}+(p^{s_1}-1)\mathcal{G}^2}{p^{s_1}}.
\end{aligned}
\end{equation*}
\end{proof}
\par
By the orthogonal property of additive characters, we can easily verify the following result.
\begin{lemma}\label{le3.2}
Let
$$\Omega_2=\sum_{z_2\in\mathbb{F}_{p^{s_2}}^*}\chi''_{z_2}(c)\sum_{x,y\in\mathbb{F}_{p^s}}\chi_{z_2}(ax+by),$$
where $a, b\in\mathbb{F}_{p^s}$ and $c\in\mathbb{F}_{p^{s_2}}$. Then
$$\Omega_2=
\left\{
  \begin{array}{ll}
    (p^{s_2}-1)p^{2s}, & \hbox{if $a=b=0$, $c=0$;} \\
    -p^{2s}, & \hbox{if $a=b=0$, $c\neq0$;} \\
    0, & \hbox{otherwise.}
  \end{array}
\right.
$$
\end{lemma}
\begin{lemma}\label{le3.3}
Let $\Gamma(s_i)={\rm Tr}_{s_i}^s(a^2+b^2)$ for $i=1, 2$ and $a,b\in\mathbb{F}_{p^s}$. Let
$$\Omega_3=\sum_{z_1\in\mathbb{F}_{p^{s_1}}^*}\sum_{z_2\in\mathbb{F}_{p^{s_2}}^*}\chi_{z_2}''(c)\sum_{x\in\mathbb{F}_{p^s}}\chi(z_1x^2+z_2ax)
\sum_{y\in\mathbb{F}_{p^s}}\chi(z_1y^2+z_2by).$$
Then the following results hold.
\par
$({\rm i})$~If $s_2|s_1$, then
$$\Omega_3=
\left\{
  \begin{array}{ll}
    (p^{s_1}-1)(p^{s_2}-1)\mathcal{G}^2, & \hbox{if~~$\Gamma(s_1)=0,c=0$;} \\
    -(p^{s_1}-1)\mathcal{G}^2, & \hbox{if~~$\Gamma(s_1)=0,c\neq0$;} \\
    -(p^{s_2}-1)\mathcal{G}^2, & \hbox{if~~$\Gamma(s_1)\neq0,c=0$;} \\
    \mathcal{G}^2, & \hbox{if~~$\Gamma(s_1)\neq0,c\neq0$.}
  \end{array}
\right.
$$
\par
$({\rm ii})$~If $s_1|s_2$ and let $\Theta={\rm Tr}_{s_1}^{s_2}\left(\frac{c^2}{\Gamma(s_2)}\right)$ for $c\in\mathbb{F}_{p^{s_2}}$, then
$$\Omega_3=
\left\{
  \begin{array}{ll}
    (p^{s_1}-1)(p^{s_2}-1)\mathcal{G}^2, & \hbox{if~~$\Gamma(s_2)=0$, $c=0$;} \\
    -(p^{s_1}-1)\mathcal{G}^2, & \hbox{\tabincell{l}{if~~$\Gamma(s_2)=0$, $c\neq0$; or~~\\$\Gamma(s_2)\neq0$, $c=0$,$\frac{s_2}{s_1}$ is odd; or\\$\Gamma(s_2)\neq0$, $c\neq0$, $\frac{s_2}{s_1}$ is odd, $\Theta=0$;}} \\
    \left(\eta''(-\Gamma(s_2))\mathcal{G}''-1\right)(p^{s_1}-1)\mathcal{G}^2, & \hbox{\tabincell{l}{if~~$\Gamma(s_2)\neq0$, $c=0$,$\frac{s_2}{s_1}$ is even; or\\$\Gamma(s_2)\neq0$, $c\neq0$, $\frac{s_2}{s_1}$ is even, $\Theta=0$;}}\\
    (\eta''(-\Gamma(s_2))\eta'(\Theta)\mathcal{G}'\mathcal{G}''-p^{s_1}+1)\mathcal{G}^2,& \hbox{if~~$\Gamma(s_2)\neq0$, $c\neq0$, $\frac{s_2}{s_1}$ is odd, $\Theta\neq0$;} \\
     (-\eta''(-\Gamma(s_2))\mathcal{G}''-p^{s_1}+1)\mathcal{G}^2,& \hbox{if~~$\Gamma(s_2)\neq0$, $c\neq0$, $\frac{s_2}{s_1}$ is even, $\Theta\neq0$.} \\
  \end{array}
\right.
$$
\end{lemma}
\begin{proof}
By Lemma \ref{le2.3}, we have
\begin{equation*}
\begin{aligned}
\Omega_3&=\mathcal{G}^2\sum_{z_2\in\mathbb{F}_{p^{s_2}}^*}\chi_{z_2}''(c)\sum_{z_1\in\mathbb{F}_{p^{s_1}}^*}\chi\left(-\frac{z_2^2}{4z_1}(a^2+b^2)\right).
\end{aligned}
\end{equation*}
To obtain the value of $\Omega_3$, we need to consider the following two cases.
\par
\textbf{Case 1}:~$s_2|s_1$. In this case, we have $-\frac{z_2^2}{4z_1}\in\mathbb{F}_{p^{s_1}}^*$ and
\begin{equation*}
\begin{aligned}
\Omega_3&=\mathcal{G}^2\sum_{z_2\in\mathbb{F}_{p^{s_2}}^*}\chi_{z_2}''(c)\sum_{z_1\in\mathbb{F}_{p^{s_1}}^*}\chi'\left(-\frac{z_2^2}{4z_1}\Gamma(s_1)\right).\\
\end{aligned}
\end{equation*}
If $\Gamma(s_1)=0$, then
$$\Omega_3=\mathcal{G}^2(p^{s_1}-1)\sum_{z_2\in\mathbb{F}_{p^{s_2}}^*}\chi_{z_2}''(c)=
\left\{
  \begin{array}{ll}
    (p^{s_1}-1)(p^{s_2}-1)\mathcal{G}^2, & \hbox{if $c=0$;} \\
    -(p^{s_1}-1)\mathcal{G}^2, & \hbox{if $c\neq0$.}
  \end{array}
\right.
$$
If $\Gamma(s_1)\neq0$, then
$$\Omega_3=-\mathcal{G}^2\sum_{z_2\in\mathbb{F}_{p^{s_2}}^*}\chi_{z_2}''(c)=
\left\{
  \begin{array}{ll}
    -(p^{s_2}-1)\mathcal{G}^2, & \hbox{if $c=0$;} \\
    \mathcal{G}^2, & \hbox{if $c\neq0$.}
  \end{array}
\right.
$$
\par
\textbf{Case~2}:~$s_1|s_2$. In this case, we have $-\frac{z_2^2}{4z_1}\in\mathbb{F}_{p^{s_2}}^*$ and
\begin{equation*}
\begin{aligned}
\Omega_3&=\mathcal{G}^2\sum_{z_1\in\mathbb{F}_{p^{s_1}}^*}\sum_{z_2\in\mathbb{F}_{p^{s_2}}^*}\chi''\left(z_2c-\frac{z_2^2}{4z_1}\Gamma(s_2)\right).\\
\end{aligned}
\end{equation*}

If $\Gamma(s_2)=0$, then
$$\Omega_3=\mathcal{G}^2\sum_{z_1\in\mathbb{F}_{p^{s_1}}^*}\sum_{z_2\in\mathbb{F}_{p^{s_2}}^*}\chi_{z_2}''(c)=
\left\{
  \begin{array}{ll}
    (p^{s_1}-1)(p^{s_2}-1)\mathcal{G}^2, & \hbox{if~~$c=0$;} \\
    -(p^{s_1}-1)\mathcal{G}^2, & \hbox{if~~$c\neq0$.}
  \end{array}
\right.
$$

If $\Gamma(s_2)\neq0$, then
\begin{equation}\label{Omega3}
\begin{aligned}
\Omega_3&=\mathcal{G}^2\sum_{z_1\in\mathbb{F}_{p^{s_1}}^*}\sum_{z_2\in\mathbb{F}_{p^{s_2}}}\chi''\left(z_2c-\frac{z_2^2}{4z_1}\Gamma(s_2)\right)-(p^{s_1}-1)\mathcal{G}^2\\
&=\eta''\left(-\Gamma(s_2)\right)\mathcal{G}''\mathcal{G}^2\sum_{z_1\in\mathbb{F}_{p^{s_1}}^*}\chi''\left(\frac{z_1c^2}{\Gamma(s_2)}\right)
\eta''(\frac{1}{4z_1})-(p^{s_1}-1)\mathcal{G}^2.\\
\end{aligned}
\end{equation}
If $c=0$, by Eq.(\ref{Omega3}) we have
\begin{equation*}
\begin{aligned}
\Omega_3&=
\left\{
  \begin{array}{ll}
    \left(\eta''\left(-\Gamma(s_2)\right)\mathcal{G}''\sum_{z_1\in\mathbb{F}_{p^{s_1}}^*}\eta'(\frac{1}{4z_1})-p^{s_1}+1\right)\mathcal{G}^2, & \hbox{if $\frac{s_2}{s_1}$ is odd;} \\
    \left(\eta''\left(-\Gamma(s_2)\right)\mathcal{G}''\sum_{z_1\in\mathbb{F}_{p^{s_1}}^*}1-p^{s_1}+1\right)\mathcal{G}^2, & \hbox{if $\frac{s_2}{s_1}$ is even.}
  \end{array}
\right.\\
&=\left\{
    \begin{array}{ll}
      -(p^{s_1}-1)\mathcal{G}^2, & \hbox{if $\frac{s_2}{s_1}$ is odd;} \\
      (\eta''\left(-\Gamma(s_2)\right)\mathcal{G}''-1)(p^{s_1}-1)\mathcal{G}^2, & \hbox{if  $\frac{s_2}{s_1}$ is even.}
    \end{array}
  \right.
\end{aligned}
\end{equation*}
If $c\neq0$, it is easy to check that $\chi''\left(\frac{z_1c^2}{\Gamma(s_2)}\right)=\chi_{z_1}'(\Theta)$. Then from Eq.(\ref{Omega3}) we obtain
\begin{equation}\label{Omega31}
\begin{aligned}
\Omega_3&=
\left\{
  \begin{array}{ll}
    \eta''\left(-\Gamma(s_2)\right)\mathcal{G}''\mathcal{G}^2\sum_{z_1\in\mathbb{F}_{p^{s_1}}^*}\chi_{z_1}'(\Theta)\eta'(\frac{1}{4z_1})-(p^{s_1}-1)\mathcal{G}^2, & \hbox{if $\frac{s_2}{s_1}$ is odd;} \\
    \eta''\left(-\Gamma(s_2)\right)\mathcal{G}''\mathcal{G}^2\sum_{z_1\in\mathbb{F}_{p^{s_1}}^*}\chi_{z_1}'(\Theta)-(p^{s_1}-1)\mathcal{G}^2, & \hbox{if $\frac{s_2}{s_1}$ is even.}
  \end{array}
\right.
\end{aligned}
\end{equation}
If $\Theta=0$, Eq.(\ref{Omega31}) can be written as
\begin{equation*}
\begin{aligned}
\Omega_3&=
\left\{
  \begin{array}{ll}
    -(p^{s_1}-1)\mathcal{G}^2, & \hbox{if $\frac{s_2}{s_1}$ is odd;} \\
    (\eta''\left(-\Gamma(s_2)\right)\mathcal{G}''-1)(p^{s_1}-1)\mathcal{G}^2, & \hbox{if $\frac{s_2}{s_1}$ is even.}
  \end{array}
\right.
\end{aligned}
\end{equation*}
If $\Theta\neq0$, let $Y=z_1\Theta\in\mathbb{F}_{p^{s_1}}$, then $$\sum_{z_1\in\mathbb{F}_{p^{s_1}}^*}\chi''\left(\frac{z_1c^2}{\Gamma(s_2)}\right)
\eta''(\frac{1}{4z_1})=\sum_{z_1\in\mathbb{F}_{p^{s_1}}^*}\chi_{z_1}'(\Theta)\eta'(\frac{1}{4z_1})=\eta'(\Theta)\sum_{Y\in\mathbb{F}_{p^{s_1}}^*}
\chi'(Y)\eta'(Y^{-1})=\eta'(\Theta)\mathcal{G}'$$
for $\frac{s_2}{s_1}$ being odd, and
$$\sum_{z_1\in\mathbb{F}_{p^{s_1}}^*}\chi''\left(\frac{z_1c^2}{\Gamma(s_2)}\right)
\eta''(\frac{1}{4z_1})=\sum_{z_1\in\mathbb{F}_{p^{s_1}}^*}\chi_{z_1}'(\Theta)=\sum_{Y\in\mathbb{F}_{p^{s_1}}^*}\chi'(Y)$$
for $\frac{s_2}{s_1}$ being even.
Hence, from Eq.(\ref{Omega31}) we have
\begin{equation*}
\begin{aligned}
\Omega_3=\left\{
    \begin{array}{ll}
      \eta''(-\Gamma(s_2))\eta'(\Theta)\mathcal{G}'\mathcal{G}''\mathcal{G}^2-(p^{s_1}-1)\mathcal{G}^2, & \hbox{if $c\neq0$ and $\frac{s_2}{s_1}$ is odd;} \\
      -\eta''\left(-\Gamma(s_2)\right)\mathcal{G}''G^2-(p^{s_1}-1)\mathcal{G}^2, & \hbox{if $c\neq0$ and $\frac{s_2}{s_1}$ is even}
    \end{array}
  \right.
\end{aligned}
\end{equation*}
for $\Theta\neq0$.

Combining all the cases, the desired conclusion follows.
\end{proof}

\begin{lemma}\label{le3.4}
Let $\frac{s_2}{s_1}$ be odd and
$$N(c,\rho)=|\left\{c\in\mathbb{F}_{p^{s_2}}^*:{\rm Tr}_{s_1}^{s_2}(\mu^{-1}c^2)=\rho\right\}|,$$
where $\mu\in\mathbb{F}_{p^{s_2}}^*$ and $\rho\in\mathbb{F}_{p^{s_1}}$. Then
$$
N(c,\rho)=\left\{
            \begin{array}{ll}
              p^{s_2-s_1}-1, & \hbox{if $\rho=0$;} \\
              p^{s_2-s_1}+\frac{\eta''(\mu)\eta'(-\rho)\mathcal{G}'\mathcal{G}''}{p^{s_1}}, & \hbox{if $\rho\neq0$.} \\
            \end{array}
          \right.
$$
\end{lemma}
\begin{proof}
By the orthogonal relation of additive characters and Lemma \ref{le2.3},
\begin{equation*}
\begin{aligned}
N(c,\rho)&=\frac{1}{p^{s_1}}\sum_{c\in\mathbb{F}_{p^{s_2}}^*}\sum_{{z_1}\in\mathbb{F}_{p^{s_1}}}\chi_{z_1}'({\rm Tr}_{s_1}^{s_2}(\mu^{-1}c^2)-\rho)\\
&=\frac{1}{p^{s_1}}\sum_{c\in\mathbb{F}_{p^{s_2}}^*}(1+\sum_{{z_1}\in\mathbb{F}_{p^{s_1}}^*}\chi_{z_1}''(\mu^{-1}c^2)\chi_{z_1}'(-\rho))\\
&=\frac{p^{s_2}-1}{p^{s_1}}+\frac{1}{p^{s_1}}\sum_{{z_1}\in\mathbb{F}_{p^{s_1}}^*}\chi_{z_1}'(-\rho)\sum_{c\in\mathbb{F}_{p^{s_2}}}\chi_{z_1}''(\mu^{-1}c^2)-
\frac{1}{p^{s_1}}\sum_{{z_1}\in\mathbb{F}_{p^{s_1}}^*}\chi_{z_1}'(-\rho)\\
&=\frac{p^{s_2}-1}{p^{s_1}}+\frac{\mathcal{G}''}{p^{s_1}}\sum_{{z_1}\in\mathbb{F}_{p^{s_1}}^*}\chi_{z_1}'(-\rho)\eta''(z_1\mu^{-1})-
\frac{1}{p^{s_1}}\sum_{{z_1}\in\mathbb{F}_{p^{s_1}}^*}\chi_{z_1}'(-\rho)\\
&=\frac{p^{s_2}-1}{p^{s_1}}+\frac{\eta''(\mu)\mathcal{G}''}{p^{s_1}}\sum_{{z_1}\in\mathbb{F}_{p^{s_1}}^*}\chi_{z_1}'(-\rho)\eta'(z_1)-
\frac{1}{p^{s_1}}\sum_{{z_1}\in\mathbb{F}_{p^{s_1}}^*}\chi_{z_1}'(-\rho).\\
\end{aligned}
\end{equation*}

If $\rho=0$, then
$\sum_{{z_1}\in\mathbb{F}_{p^{s_1}}^*}\chi_{z_1}'(-\rho)\eta'(z_1)=0~{\rm and}~\sum_{{z_1}\in\mathbb{F}_{p^{s_1}}^*}\chi_{z_1}'(-\rho)=p^{s_1}-1$. This implies that thus $N(c,\rho)=p^{s_2-s_1}-1$.

If $\rho\neq0$, then $\sum_{{z_1}\in\mathbb{F}_{p^{s_1}}^*}\chi_{z_1}'(-\rho)=-1$ and
\begin{equation*}
\begin{aligned}
\sum_{{z_1}\in\mathbb{F}_{p^{s_1}}^*}\chi_{z_1}'(-\rho)\eta'(z_1)&=\sum_{{z_1}\in\mathbb{F}_{p^{s_1}}}\chi_{z_1}'(-\rho)\eta'(z_1)=
\sum_{{z_2}\in\mathbb{F}_{p^{s_1}}}\chi'(z_2)\eta'(z_2)\eta'(-\rho^{-1})=\eta'(-\rho)\mathcal{G}',
\end{aligned}
\end{equation*}
where $z_2=-z_1\rho$. Hence, $N(c,\rho)=p^{s_2-s_1}+\frac{\eta''(\mu)\eta'(-\rho)\mathcal{G}'\mathcal{G}''}{p^{s_1}}$.
\end{proof}

\begin{lemma}\label{le3.5}
Follow the notation in Lemma \ref{le3.3}. Let
$$N(a,b)=|\left\{(x,y)\in(\mathbb{F}_{p^s},\mathbb{F}_{p^s}): {\rm Tr}_{s_1}^s(x^{2}+y^{2})=0~{\rm and}~{\rm Tr}_{s_2}^s(ax+by)+c=0\right\}|,$$
where $a,b\in\mathbb{F}_{p^s}$ and $c\in\mathbb{F}_{p^{s_2}}$. Then the following results hold.

$({\rm i})$~If $s_2|s_1$ (or $s_1|s_2$) and $(a,b)=(0,0)$, then
$$N(a,b)=
\left\{
  \begin{array}{ll}
    p^{2s-s_1}+p^{-s_1}(p^{s_1}-1)\mathcal{G}^2, & \hbox{if $c=0$;} \\
    0, & \hbox{if $c\neq0$.}
  \end{array}
\right.
$$
\par
$({\rm ii})$~If $s_2|s_1$ and $(a,b)\neq(0,0)$, then
$$
N(a,b)=\left\{
         \begin{array}{ll}
           p^{2s-(s_1+s_2)}+p^{-s_1}(p^{s_1}-1)\mathcal{G}^2, & \hbox{if $\Gamma(s_1)=0$, $c=0$;} \\
           p^{2s-(s_1+s_2)}, & \hbox{if $\Gamma(s_1)=0$, $c\neq0$;} \\
           p^{-(s_1+s_2)}(p^{2s}+(p^{s_1}-p^{s_2})\mathcal{G}^2), & \hbox{if $\Gamma(s_1)\neq0$, $c=0$;} \\
           p^{2s-(s_1+s_2)}+p^{-s_2}\mathcal{G}^2, & \hbox{if $\Gamma(s_1)\neq0$, $c\neq0$.}
         \end{array}
       \right.
$$
\par
$({\rm iii})$~If $s_1|s_2$ and $(a,b)\neq(0,0)$, then
$$
N(a,b)=\left\{
         \begin{array}{ll}
           p^{2s-(s_1+s_2)}+p^{-s_1}(p^{s_1}-1)\mathcal{G}^2, & \hbox{if $\Gamma(s_2)=0$, $c=0$;} \\
           p^{2s-(s_1+s_2)},& \hbox{\tabincell{l}{if $\Gamma(s_2)=0$, $c\neq0$; or $\Gamma(s_2)\neq0$, $c=0$, $\frac{s_2}{s_1}$\\is odd; or $\Gamma(s_2)\neq0$, $c\neq0$, $\frac{s_2}{s_1}$ is odd, $\Theta=0$;}} \\
           p^{-(s_1+s_2)}\left(p^{2s}+\eta''(-\Gamma(s_2))\eta'(\Theta)\mathcal{G}'\mathcal{G}''\mathcal{G}^2\right), & \hbox{if $\Gamma(s_2)\neq0$, $c\neq0$, $\frac{s_2}{s_1}$ is odd, $\Theta\neq0$;} \\
         p^{-(s_1+s_2)}\left(p^{2s}+\eta''(-\Gamma(s_2))(p^{s_1}-1)\mathcal{G}''\mathcal{G}^2\right), & \hbox{\tabincell{l}{if $\Gamma(s_2)\neq0$, $c=0$, $\frac{s_2}{s_1}$ is even; or\\$\Gamma(s_2)\neq0$, $c\neq0$, $\frac{s_2}{s_1}$ is even, $\Theta=0$;}}\\
         p^{-(s_1+s_2)}\left(p^{2s}-\eta''(-\Gamma(s_2))\mathcal{G}''\mathcal{G}^2\right), & \hbox{if $\Gamma(s_2)\neq0$, $c\neq0$, $\frac{s_2}{s_1}$ is even, $\Theta\neq0$.}
         \end{array}
       \right.
$$
\end{lemma}
\begin{proof}
By the orthogonal relation of additive characters, we have
\begin{equation*}
\begin{aligned}
N(a,b)&=\sum_{x,y\in\mathbb{F}_{p^s}}\left(\frac{1}{p^{s_1}}\sum_{z_1\in\mathbb{F}_{p^{s_1}}}\zeta_p^{z_1{\rm Tr}_{1}^{s_1}({\rm Tr}_{s_1}^s(x^2+y^2))}\right)
\left(\frac{1}{p^{s_2}}\sum_{z_2\in\mathbb{F}_{p^{s_2}}}\zeta_p^{z_2{\rm Tr}_{1}^{s_2}({\rm Tr}_{s_2}^s(ax+by)+c)}\right)\\
&=\frac{1}{p^{s_1+s_2}}\sum_{x,y\in\mathbb{F}_{p^s}}\left(1+\sum_{z_1\in\mathbb{F}_{p^{s_1}}^*}\chi_{z_1}(x^2+y^2)\right)
\left(1+\sum_{z_2\in\mathbb{F}_{p^{s_2}}^*}\chi_{z_2}(ax+by)\chi_{z_2}''(c)\right)\\
&=p^{2s-s_1-s_2}+\frac{1}{p^{s_1+s_2}}(\Omega_1+\Omega_2+\Omega_3),
\end{aligned}
\end{equation*}
where $$\Omega_1=\sum_{z_1\in\mathbb{F}_{p^{s_1}}^*}\sum_{x,y\in\mathbb{F}_{p^s}}\chi_{z_1}(x^2+y^2), ~~\Omega_2=\sum_{z_2\in\mathbb{F}_{p^{s_2}}^*}\chi_{z_2}''(c)\sum_{x,y\in\mathbb{F}_{p^s}}\chi_{z_2}(ax+by)$$ and $$\Omega_3=\sum_{z_1\in\mathbb{F}_{p^{s_1}}^*}\sum_{z_2\in\mathbb{F}_{p^{s_2}}^*}\chi_{z_2}''(c)\sum_{x\in\mathbb{F}_{p^s}}\chi(z_1x^2+z_2ax)
\sum_{y\in\mathbb{F}_{p^s}}\chi(z_1y^2+z_2by).$$
Hence, the value of $N(a,b)$ can be determined by Lemmas \ref{le3.1}, \ref{le3.2} and \ref{le3.3}.
\end{proof}

\section{The main results and their proofs}
In this section, we present the weight distribution of the linear code $\mathcal{C}_D$ defined by Eq.(\ref{eq1.4}) and prove that $\mathcal{C}_D$ is self-orthogonal under certain conditions. Furthermore, the parameters of the dual of $\mathcal{C}_D$ and a class of AMDS codes are obtained. In addition, we show that $\mathcal{C}_D$ is a locally recoverable code.

\begin{theorem}\label{th4.1}
Let $p$ be an odd prime and $s$ be a positive integer. Let $s_1$ and $s_2$ be two positive divisors of $s$. Let $\mathcal{C}_D$ be the linear code defined by Eq.(\ref{eq1.4}). Then the following statements hold.
\par
$({\rm i})$~If $s_2|s_1$, then $\mathcal{C}_D$  is a $\left[\frac{p^{2s}+(p^{s_1}-1)\mathcal{G}^2}{p^{s_1}},\frac{2s}{s_2}+1\right]_{p^{s_2}}$ linear code with weight distribution in Table~\ref{tab1}. Moreover, if $s\geq 2s_1$, then $\mathcal{C}_D$ is self-orthogonal.
\par
$({\rm ii})$~If $\frac{s_2}{s_1}$ is odd, then $\mathcal{C}_D$ is a $\left[\frac{p^{2s}+(p^{s_1}-1)\mathcal{G}^2}{p^{s_1}},\frac{2s}{s_2}+1\right]_{p^{s_2}}$ linear code with weight distribution in Table~\ref{tab2}. Moreover, if $2s>s_1+s_2$, then $\mathcal{C}_D$ is self-orthogonal.
\par
$({\rm iii})$~If $\frac{s_2}{s_1}$ is even, then $\mathcal{C}_D$ is a $\left[\frac{p^{2s}+(p^{s_1}-1)\mathcal{G}^2}{p^{s_1}},\frac{2s}{s_2}+1\right]_{p^{s_2}}$ linear code with weight distribution in Table~\ref{tab3}. Moreover, if $2s>2s_1+s_2$, then $\mathcal{C}_D$ is self-orthogonal.
\par
$({\rm iv})$~The dual $\mathcal{C}_D^\bot$ of $\mathcal{C}_D$ is a $\left[\frac{p^{2s}+(p^{s_1}-1)\mathcal{G}^2}{p^{s_1}},\frac{p^{2s}+(p^{s_1}-1)\mathcal{G}^2}{p^{s_1}}-(\frac{2s}{s_2}+1),3\right]_{p^{s_2}}$ linear code. Moreover, if $s=s_2$, then $\mathcal{C}_D^\bot$ is an AMDS code.
\end{theorem}
\begin{table}[h!]\label{tab1}
\tiny
\setlength{\tabcolsep}{5mm}
\caption{The weight distribution of $\mathcal{C}_D$ given in Theorem \ref{th4.1} ${\rm(i)}$.}
\centering
\label{tab1}
 \begin{tabular}{cc}
\toprule
Weight&$A_{w}$\\
\midrule
$0$&$1$\\
$\frac{p^{2s}+(p^{s_1}-1)\mathcal{G}^2}{p^{s_1}}$&$p^{s_2}-1$\\
$\frac{p^{2s}(p^{s_2}-1)}{p^{s_1+s_2}}$&$\frac{p^{2s}+(p^{s_1}-1)\mathcal{G}^2}{p^{s_1}}-1$\\
$\frac{p^{2s}(p^{s_2}-1)}{p^{s_1+s_2}}+\frac{(p^{s_1}-1)\mathcal{G}^2}{p^{s_1}}$&$(p^{s_2}-1)(\frac{p^{2s}+(p^{s_1}-1)\mathcal{G}^2}{p^{s_1}}-1)$\\
$\frac{p^{2s}(p^{s_2}-1)}{p^{s_1+s_2}}+\frac{(p^{s_1+s_2}-p^{s_1})\mathcal{G}^2}{p^{s_1+s_2}}$&$\frac{(p^{s_1}-1)(p^{2s}-\mathcal{G}^2)}{p^{s_1}}$\\
$\frac{p^{2s}(p^{s_2}-1)}{p^{s_1+s_2}}+\frac{(p^{s_1+s_2}-p^{s_1}-p^{s_2})\mathcal{G}^2}{p^{s_1+s_2}}$&$\frac{(p^{s_2}-1)(p^{s_1}-1)(p^{2s}-\mathcal{G}^2)}{p^{s_1}}$\\
\bottomrule
 \end{tabular}
\end{table}
\begin{table}[h!]\label{tab2}
\tiny
\setlength{\tabcolsep}{3mm}
\caption{The weight distribution of $\mathcal{C}_D$ given in Theorem \ref{th4.1} ${\rm(ii)}$.}
\centering
\label{tab2}
 \begin{tabular}{cc}
\toprule
Weight&$A_{w}$\\
\midrule
$0$&$1$\\
$\frac{p^{2s}+(p^{s_1}-1)\mathcal{G}^2}{p^{s_1}}$&$p^{s_2}-1$\\
$\frac{p^{2s}(p^{s_2}-1)}{p^{s_1+s_2}}$&$\frac{p^{2s}+(p^{s_2}-1)\mathcal{G}^2}{p^{s_2}}-1$\\
$\frac{p^{2s}(p^{s_2}-1)}{p^{s_1+s_2}}+\frac{(p^{s_1}-1)\mathcal{G}^2}{p^{s_1}}$&
$\frac{p^{s_2}(p^{2s}-1)}{p^{s_1}}+(\frac{p^{s_1+s_2}-p^{s_2}}{p^{s_1}}-1)(\frac{p^{2s}+(p^{s_2}-1)\mathcal{G}^2}{p^{s_2}}-1)$\\
$\frac{p^{2s}+(p^{s_1}-1)\mathcal{G}^2}{p^{s_1}}-\frac{p^{2s}+\eta''(-1)\mathcal{G}'\mathcal{G}''\mathcal{G}^2}{p^{s_1+s_2}}$&
$\frac{p^{s_1}-1}{2}(p^{2s}-\frac{p^{2s}+(p^{s_2}-1)\mathcal{G}^2}{p^{s_2}})(p^{s_2-s_1}+\frac{\eta'(-1)\mathcal{G}'\mathcal{G}''}{p^{s_1}})$\\
$\frac{p^{2s}+(p^{s_1}-1)\mathcal{G}^2}{p^{s_1}}-\frac{p^{2s}-\eta''(-1)\mathcal{G}'\mathcal{G}''\mathcal{G}^2}{p^{s_1+s_2}}$&
$\frac{p^{s_1}-1}{2}(p^{2s}-\frac{p^{2s}+(p^{s_2}-1)\mathcal{G}^2}{p^{s_2}})(p^{s_2-s_1}-\frac{\eta'(-1)\mathcal{G}'\mathcal{G}''}{p^{s_1}})$\\
\bottomrule
 \end{tabular}
\end{table}
\begin{table}[h!]\label{tab3}
\tiny
\setlength{\tabcolsep}{3mm}
\caption{The weight distribution of $\mathcal{C}_D$ given in Theorem \ref{th4.1} ${\rm(iii)}$.}
\centering
\label{tab3}
 \begin{tabular}{cc}
\toprule
Weight&$A_{w}$\\
\midrule
$0$&$1$\\
$\frac{p^{2s}+(p^{s_1}-1)\mathcal{G}^2}{p^{s_1}}$&$p^{s_2}-1$\\
$\frac{p^{2s}(p^{s_2}-1)}{p^{s_1+s_2}}$&$\frac{p^{2s}+(p^{s_2}-1)\mathcal{G}^2}{p^{s_2}}-1$\\
$\frac{p^{2s}(p^{s_2}-1)}{p^{s_1+s_2}}+\frac{(p^{s_1}-1)\mathcal{G}^2}{p^{s_1}}$&$(p^{s_2}-1)\left(\frac{p^{2s}+(p^{s_2}-1)\mathcal{G}^2}{p^{s_2}}-1\right)$\\
$\frac{p^{2s}+(p^{s_1}-1)\mathcal{G}^2}{p^{s_1}}-\frac{p^{2s}+\eta''(-1)(p^{s_1}-1)\mathcal{G}''\mathcal{G}^2}{p^{s_1+s_2}}$&
$\frac{(p^{s_2}-1)(p^{2s}-\mathcal{G}^2)(p^{s_2}+(p^{s_1}-1)\mathcal{G}'')}{2p^{s_1+s_2}}$\\
$\frac{p^{2s}+(p^{s_1}-1)\mathcal{G}^2}{p^{s_1}}-\frac{p^{2s}-\eta''(-1)(p^{s_1}-1)\mathcal{G}''\mathcal{G}^2}{p^{s_1+s_2}}$&
$\frac{(p^{s_2}-1)(p^{2s}-\mathcal{G}^2)(p^{s_2}-(p^{s_1}-1)\mathcal{G}'')}{2p^{s_1+s_2}}$\\
$\frac{p^{2s}+(p^{s_1}-1)\mathcal{G}^2}{p^{s_1}}-\frac{p^{2s}-\eta''(-1)\mathcal{G}''\mathcal{G}^2}{p^{s_1+s_2}}$&
$\frac{(p^{s_1}-1)(p^{s_2}-1)(p^{2s}-\mathcal{G}^2)(p^{s_2}-\mathcal{G}'')}{2p^{s_1+s_2}}$\\
$\frac{p^{2s}+(p^{s_1}-1)\mathcal{G}^2}{p^{s_1}}-\frac{p^{2s}+\eta''(-1)\mathcal{G}''\mathcal{G}^2}{p^{s_1+s_2}}$&
$\frac{(p^{s_1}-1)(p^{s_2}-1)(p^{2s}-\mathcal{G}^2)(p^{s_2}+\mathcal{G}'')}{2p^{s_1+s_2}}$\\
\bottomrule
 \end{tabular}
\end{table}

\begin{proof}
For the first three results, we only prove the result of (ii), and the proofs of other cases are similar. By Eq.(\ref{eq1.4}), Lemmas \ref{le3.1} and \ref{le3.5}, the Hamming weight of any codeword $\mathbf{c}$ in $\mathcal{C}_D$ is as follows:
$$wt(\mathbf{c})=n-N(a,b)=\frac{p^{2s}+(p^{s_1}-1)\mathcal{G}^2}{p^{s_1}}-N(a,b),$$
where $N(a,b)$ is given in Lemma \ref{le3.5}. By Lemma \ref{le3.5}, we have the following discussions.
\par
If $(a,b)=(0,0)$ and $c=0$, then $wt(\mathbf{c})=n-(p^{2s-s_1}+\frac{p^{s_1}-1}{p^{s_1}}\mathcal{G}^2)=0$.
\par
If $a,b$ and $c$ are in the set
$$
S_1=\{a,b\in\mathbb{F}_{p^s},c\in\mathbb{F}_{p^{s_2}}:(a,b)=(0,0)~{\rm and}~c\in\mathbb{F}_{p^{s_2}}^*\},
$$
then the Hamming weights of codewords given by Eq.(\ref{eq1.4}) are $w_1=p^{2s-s_1}+p^{-s_1}(p^{s-1}-1)\mathcal{G}^2$, and $A_{w_1}=|S_1|=p^{s_2}-1$.
\par
If $a,b$ and $c$ are in the set
$$
S_2=\{a,b\in\mathbb{F}_{p^s},c\in\mathbb{F}_{p^{s_2}}:(a,b)\neq(0,0),\,\Gamma(s_2)=0~{\rm and}~c=0\},
$$
then the Hamming weights of codewords given by Eq.(\ref{eq1.4}) are
$w_2=p^{2s-s_1}+p^{-s_1}(p^{s_1}-1)\mathcal{G}^2-(p^{2s-s_1-s_2}+p^{-s_1}(p^{s-1}-1)\mathcal{G}^2)=p^{2s-s_1-s_2}(p^{s_2}-1)$, and $A_{w_2}=|S_2|$. By Lemma \ref{le3.1},
\begin{equation*}
\begin{aligned}
A_{w_2}=|\{(a,b)\in(\mathbb{F}_{p^s}^*,\mathbb{F}_{p^s}^*):{\rm Tr}_{s_2}^s(a^2+b^2)=0\}|=p^{2s-s_2}+p^{-s_2}(p^{s_2}-1)\mathcal{G}^2-1.
\end{aligned}
\end{equation*}
\par
If $a,b$ and $c$ are in the set
$$
S_3=\{a,b\in\mathbb{F}_{p^s},c\in\mathbb{F}_{p^{s_2}}: (a,b)\neq(0,0), \Gamma(s_2)=0~{\rm and}~c\in\mathbb{F}_{p^{s_2}}^*\},
$$
then the Hamming weights of codewords given by Eq.(\ref{eq1.4}) are
$w_3=p^{2s-s_1}+p^{-s_1}(p^{s_1}-1)\mathcal{G}^2-p^{2s-s_1-s_2}=p^{2s-s_1-s_2}(p^{s_2}-1)+p^{-s_1}(p^{s_1}-1)\mathcal{G}^2$, and $A_{w_3}=|S_3|$. Clearly,
\begin{equation*}
\begin{aligned}
A_{w_3}=(p^{s_2}-1)A_{w_2}=(p^{s_2}-1)\left(p^{2s-s_2}+p^{-s_2}(p^{s_2}-1)\mathcal{G}^2-1\right).
\end{aligned}
\end{equation*}
\par
If $a,b$ and $c$ are in the set
$$
S_4=\{a,b\in\mathbb{F}_{p^s},c\in\mathbb{F}_{p^{s_2}}: (a,b)\neq(0,0), \Gamma(s_2)\neq0~{\rm and}~c=0\},
$$
then the Hamming weights of codewords given by Eq.(\ref{eq1.4}) are $w_4=w_3$, and $A_{w_4}=|S_4|$. It is obvious that
$$A_{w_4}=p^{2s}-1-|S_2|=p^{2s-s_2}(p^{s_2}-1)-p^{-s_2}(p^{s_2}-1)\mathcal{G}^2.$$
\par
If $a,b$ and $c$ are in the set
$$
S_5=\{a,b\in\mathbb{F}_{p^s},c\in\mathbb{F}_{p^{s_2}}: (a,b)\neq(0,0), \Gamma(s_2)\neq0, c\in\mathbb{F}_{p^{s_2}}^*~{\rm and}~\Theta=0\},
$$
then the Hamming weights of codewords given by Eq.(\ref{eq1.4}) are $w_5=w_3$, and $A_{w_5}=|S_5|$. By Lemma~\ref{le3.4},
\begin{equation*}
\begin{aligned}
A_{w_5}=N(c,0)\cdot|S_4|&=(p^{s_2-s_1}-1)(p^{2s}-1-|S_2|)\\
&=(p^{s_2-s_1}-1)\left(p^{2s-s_2}(p^{s_2}-1)-p^{-s_2}(p^{s_2}-1)\mathcal{G}^2\right).
\end{aligned}
\end{equation*}

If $a,b$ and $c$ are in the set
$$
S_6=\{a,b\in\mathbb{F}_{p^s},c\in\mathbb{F}_{p^{s_2}}: (a,b)\neq(0,0), \Gamma(s_2)\neq0, c\in\mathbb{F}_{p^{s_2}}^*, \Theta\neq0~{\rm and}~\eta''(\Gamma(s_2))\eta'(\Theta)=1\},
$$
then the Hamming weights of codewords given by Eq.(\ref{eq1.4}) are $w_6=p^{2s-s_1}+p^{-s_1}(p^{s_1}-1)\mathcal{G}^2-p^{-(s_1+s_2)}(p^{2s}+\eta''(-1)\mathcal{G}'\mathcal{G}''\mathcal{G}^2)$.  Clearly, $\eta''(\Gamma(s_2))=\eta'(\Theta)=1$, or $\eta''(\Gamma(s_2))=\eta'(\Theta)=-1$ if $\Theta\neq0~{\rm and}~\eta''(\Gamma(s_2))\eta'(\Theta)=1$.  By Lemma \ref{le3.4},
\begin{equation*}
\begin{aligned}
A_{w_6}=&\frac{p^{s_1}-1}{2}\cdot N(c,\rho_1)\cdot|\{(a,b)\in(\mathbb{F}_{p^s}^*,\mathbb{F}_{p^s}^*):\eta''(\Gamma(s_2))=1\}|\\
       &+\frac{p^{s_1}-1}{2}\cdot N(c,\rho_2)\cdot|\{(a,b)\in(\mathbb{F}_{p^s}^*,\mathbb{F}_{p^s}^*):\eta''(\Gamma(s_2))=-1\}|\\
=&\frac{p^{s_1}-1}{2}(p^{s_2-s_1}+\eta'(-1)\mathcal{G}'\mathcal{G}''p^{-s_1})\left(p^{2s-s_2}(p^{s_2}-1)-p^{-s_2}(p^{s_2}-1)\mathcal{G}^2\right),
\end{aligned}
\end{equation*}
where $\rho_1$ and $\rho_2$ are the square and nonsquare elements of $\mathbb{F}_{p^{s_1}}$, respectively.
\par
If $a,b$ and $c$ are in the set
$$
S_7=\{a,b\in\mathbb{F}_{p^s},c\in\mathbb{F}_{p^{s_2}}: (a,b)\neq(0,0), \Gamma(s_2)\neq0, c\in\mathbb{F}_{p^{s_2}}^*, \Theta\neq0~{\rm and}~\eta''(\Gamma(s_2))\eta'(\Theta)=-1\},
$$
then the Hamming weights of codewords given by Eq.(\ref{eq1.4}) are $w_7=p^{2s-s_1}+p^{-s_1}(p^{s_1}-1)\mathcal{G}^2-p^{-(s_1+s_2)}(p^{2s}-\eta''(-1)\mathcal{G}'\mathcal{G}''\mathcal{G}^2)$.  It is obvious that $\eta''(\Gamma(s_2))=1$ and $\eta'(\Theta)=-1$, or $\eta''(\Gamma(s_2))=-1$ and $\eta'(\Theta)=1$ if $\eta''(\Gamma(s_2))\eta'(\Theta)=-1$.  By Lemma \ref{le3.4},
\begin{equation*}
\begin{aligned}
A_{w_7}=&\frac{p^{s_1}-1}{2}\cdot N(c,\rho_1)\cdot|\{(a,b)\in(\mathbb{F}_{p^s}^*,\mathbb{F}_{p^s}^*):\eta''(\Gamma(s_2))=-1\}|\\
       &+\frac{p^{s_1}-1}{2}\cdot N(c,\rho_2)\cdot|\{(a,b)\in(\mathbb{F}_{p^s}^*,\mathbb{F}_{p^s}^*):\eta''(\Gamma(s_2))=1\}|\\
=&\frac{p^{s_1}-1}{2}(p^{s_2-s_1}-\eta'(-1)\mathcal{G}'\mathcal{G}''p^{-s_1})\left(p^{2s-s_2}(p^{s_2}-1)-p^{-s_2}(p^{s_2}-1)\mathcal{G}^2\right),
\end{aligned}
\end{equation*}
where $\rho_1$ and $\rho_2$ are the square and nonsquare elements of $\mathbb{F}_{p^{s_1}}$, respectively.
\par
Hence, the weight distribution of $\mathcal{C}_D$ can be followed  from above discussions when $\frac{s_2}{s_1}$ is odd. Moreover, if $2s>s_1+s_2$, then $\mathcal{C}_D$ is $p$-divisible. Consequently, $\C_D$
is self-orthogonal by Lemma~\ref{le2.1}.

Now, we prove the result of (iv). It is easy to show that the length of $\mathcal{C}_D^\bot$ is equal to $\frac{p^{2s}+(p^{s_1}-1)\mathcal{G}^2}{p^{s_1}}$, and the dimension of $\mathcal{C}_D^\bot$ is equal to $\frac{p^{2s}+(p^{s_1}-1)\mathcal{G}^2}{p^{s_1}}-(\frac{2s}{s_2}+1)$. From the first three Pless power moment, we have
 $$A_1^\bot=A_2^\bot=0\,\, \text{and}\,\,A_3^\bot\neq 0.$$
 Hence, $d(\mathcal{C}_D^\bot)=3$. If $s=s_2$, then ${\rm dim}(\mathcal{C}_D)=3$ and ${\rm dim}(\mathcal{C}_D^\bot)=n-3$. Therefore, $\mathcal{C}_D^\bot$ is an $[n,n-3,3]_{p^{s_2}}$ AMDS code when $s=s_2$, which is optimal since the length of $\mathcal{C}_D^\bot$ is larger than $n>p^{s_2}+1$.
\end{proof}
\par
Below, we present several examples that support the results given in Theorem~\ref{th4.1}, which have been verified using Magma programs.
\begin{example}\label{ex4.3}
(1)~Let $(p,s,s_1,s_2)=(3,2,1,1)$. By Theorem \ref{th4.1} ${\rm (i)}$, $\mathcal{C}_D$ is a $[33,5,18]_3$ linear code with weight enumerator
$1+2x^{33}+112x^{24}+32x^{18}+96x^{21}$. Clearly, $\mathcal{C}_D$ is $3$-divisible and self-orthogonal by Lemma~\ref{le2.1}.
\par
(2)~Let $(p,s,s_1,s_2)=(3,3,1,1)$. By Theorem \ref{th4.1} ${\rm (ii)}$, $\mathcal{C}_D$ is a $[225,7,144]_3$ linear code with weight enumerator
$1+952x^{144}+2x^{225}+224x^{162}+1008x^{153}$. Clearly, $\mathcal{C}_D$ is $3$-divisible and self-orthogonal by Lemma~\ref{le2.1}.
\par
(3)~Let $(p,s,s_1,s_2)=(3,4,1,2)$. By Theorem \ref{th4.1} ${\rm (iii)}$, $\mathcal{C}_D$ is a $[2241,5,1944]_9$ linear code with weight enumerator
$1+14400x^{1980}+2880x^{2016}+11520x^{2007}+6400x^{1998}+800x^{1944}+8x^{2241}+23040x^{1989}$. Clearly, $\mathcal{C}_D$ is $3$-divisible and self-orthogonal by Lemma~\ref{le2.1}.
\par
(4)~Let $(p,s,s_1,s_2)=(3,2,1,2)$. By Theorem \ref{th4.1} ${\rm (iii)}$ and ${\rm (iv)}$, $\mathcal{C}_D$ is a $[33,3,24]_9$ linear code with weight enumerator
$1+8x^{33}+16x^{24}+160x^{28}+256x^{29}+128x^{30}+128x^{31}+32x^{32}$.  Besides, its dual $\mathcal{C}_D^\bot$ is a $[33,30,3]_9$ AMDS code.
\end{example}
\begin{remark}\label{re4.3}
Let $\mathbb{F}_{p^s}^*=\langle\omega\rangle$ and $\{1,\omega,\ldots,\omega^{\frac{s}{s_2}-1}\}$ be the basis of $\mathbb{F}_{p^s}$ over $\mathbb{F}_{p^{s_2}}$.
Since the dimension of the linear code $\mathcal{C}_D$ defined by Eq.(\ref{eq1.4}) is $\frac{2s}{s_2}+1$,
by the definition, the generator matrix of  $\mathcal{C}_D$ can be represented as follows:
\begin{equation}\label{matrix:{G_D}}
\begin{aligned}
G_{D}=
\left[
  \begin{array}{cccc}
    1&1&\cdots &1\\
    {\rm Tr}_{s_2}^s(x_1\omega^0) & {\rm Tr}_{s_2}^s(x_2\omega^0) & \cdots & {\rm Tr}_{s_2}^s(x_n\omega^0) \\
    \vdots&\vdots&\ddots&\vdots\\
    {\rm Tr}_{s_2}^s(x_1\omega^{\frac{s}{s_2}-1}) & {\rm Tr}_{s_2}^s(x_2\omega^{\frac{s}{s_2}-1}) & \cdots & {\rm Tr}_{s_2}^s(x_n\omega^{\frac{s}{s_2}-1}) \\
    {\rm Tr}_{s_2}^s(y_1\omega^0) & {\rm Tr}_{p^s/p^{s_2}}(y_2\omega^0) & \cdots & {\rm Tr}_{s_2}^s(y_n\omega^0) \\
    \vdots&\vdots&\ddots&\vdots\\
    {\rm Tr}_{s_2}^s(y_1\omega^{\frac{s}{s_2}-1}) & {\rm Tr}_{s_2}^s(y_2\omega^{\frac{s}{s_2}-1}) & \cdots & {\rm Tr}_{s_2}^s(y_n\omega^{\frac{s}{s_2}-1}) \\
  \end{array}
\right],
\end{aligned}
\end{equation}
where $D=\{(x_1,y_1),(x_2,y_2)\ldots,(x_n,y_n)=(0,0)\}$. Let $G_D=[\mathbf{g}_1,\mathbf{g}_2,\ldots,\mathbf{g}_n]$, where $\mathbf{g}_i$ denotes the $i$-th column vector of $G_D$ for $1\leq i\leq n$. Using the method given in \cite[Theorem 3.2]{Heng2024}, for any column vector $\mathbf{g}_i$ of $G_D$, there exists a column vector $\mathbf{g}_j$ of $G_D$ such that $\mathbf{g}_n$ can be linearly represented by $\mathbf{g}_i$ and $\mathbf{g}_j$ over $\mathbb{F}_{p^{s_2}}$, where $1\leq i,j\leq n-1,i\neq j$. From \cite{Huang2016}, we know that the linear code $\mathcal{C}_D$ defined by Eq.(\ref{eq1.4}) is a locally recoverable code with parameters
$$[n,k,d;r]_q=\left[\frac{p^{2s}+(p^{s_1}-1)\mathcal{G}^2}{p^{s_1}},\frac{2s}{s_2}+1,d(\mathcal{C}_D);2\right]_{p^{s_2}}.$$
Next, we provide an example of locally recoverable code as follows. Let $(p,s,s_1,s_2)=(3,2,2,1)$ and $\omega$ be the primitive element of $\mathbb{F}_{3^2}$. Then $\mathcal{C}_{D}$ defined by Eq.(\ref{eq1.4}) is a $[17,5,6]_3$ linear code with defining set
\begin{align*}
\setlength{\arraycolsep}{1.2pt}
\begin{array}{lc}
D=
     \begin{Bmatrix}
\begin{array}{ccccccccc}
  (1,\omega^2),&(1,\omega^6),&(\omega,\omega^3),&(\omega,\omega^7),&(\omega^2,1),&(\omega^2,2),&(\omega^3,\omega),&(\omega^3,\omega^5),&(2,\omega^2),\\
  (2,\omega^6),&(\omega^5,\omega^3),&(\omega^5,\omega^7),&(\omega^6,1),&(\omega^6,2),&(\omega^7,\omega),&(\omega^7,\omega^5),&(0,0)
\end{array}
\end{Bmatrix}.
 \end{array}
\end{align*}
By Magma programs, its generator matrix is
\begin{equation*}
\begin{aligned}
G_{D}=
\left[
  \begin{array}{ccccccccccccccccc}
    1 & 1 & 1 & 1 & 1 & 1 & 1 & 1 & 1 & 1 & 1 & 1 & 1 & 1 & 1 & 1 & 1 \\
    2 & 2 & 1 & 1 & 0 & 0 & 1 & 1 & 1 & 1 & 2 & 2 & 0 & 0 & 2 & 2 & 0 \\
    1 & 1 & 0 & 0 & 1 & 1 & 1 & 1 & 2 & 2 & 0 & 0 & 2 & 2 & 2 & 2 & 0 \\
    0 & 0 & 1 & 2 & 2 & 1 & 1 & 2 & 0 & 0 & 1 & 2 & 2 & 1 & 1 & 2 & 0 \\
    1 & 2 & 1 & 2 & 1 & 2 & 0 & 0 & 1 & 2 & 1 & 2 & 1 & 2 & 0 & 0 & 0 \\
  \end{array}
\right].
\end{aligned}
\end{equation*}
Let $\mathbf{g}_i$ denote the $i$-th column vector of $G_D$ for $1\leq i\leq 17$. Then
$$\mathbf{g}_{17}=2\mathbf{g}_1+2\mathbf{g}_{10},\mathbf{g}_{17}=2\mathbf{g}_2+2\mathbf{g}_{9},\mathbf{g}_{17}=2\mathbf{g}_3+2\mathbf{g}_{12},
\mathbf{g}_{17}=2\mathbf{g}_4+2\mathbf{g}_{11},$$
$$\mathbf{g}_{17}=2\mathbf{g}_5+2\mathbf{g}_{14},\mathbf{g}_{17}=2\mathbf{g}_6+2\mathbf{g}_{13},\mathbf{g}_{17}=2\mathbf{g}_7+2\mathbf{g}_{16},
\mathbf{g}_{17}=2\mathbf{g}_8+2\mathbf{g}_{15}.$$
Hence, $\mathcal{C}_{D}$ is a locally recoverable code with parameters $[17,5,6;2]_3$.
\end{remark}
\par
At the end of this section, we list some optimal and AMDS linear codes constructed from Theorem \ref{th4.1} in Table \ref{tab4}. Notably,  we say that our codes are optimal as one can find the best known linear codes with the same parameters in the Code Tables at \cite{Grassl}.
\begin{table}[h!]
\setlength{\tabcolsep}{5mm}
\caption{Some optimal and AMDS linear codes from Theorem \ref{th4.1}.}
\centering
\label{tab4}
 \begin{tabular}{cccc}
\toprule
$(p,s,s_1,s_2)$&Code&$[n,k,d]_{p^{s_2}}$&Optimality\\
\midrule
$(3,2,1,1)$&$\mathcal{C}_D^\bot$&$[33,28,3]_{3}$&Optimal\\
$(3,2,1,2)$&$\mathcal{C}_D^\bot$&$[33,30,3]_{3^2}$&AMDS\\
$(3,2,2,2)$&$\mathcal{C}_D^\bot$&$[17,14,3]_{3^2}$&AMDS\\
$(3,3,1,1)$&$\mathcal{C}_D$&$[225,7,144]_{3}$&Optimal\\
$(3,3,1,1)$&$\mathcal{C}_D^\bot$&$[225,218,3]_{3}$&Optimal\\
$(3,3,1,3)$&$\mathcal{C}_D^\bot$&$[225,222,3]_{3^3}$&AMDS\\
$(5,1,1,1)$&$\mathcal{C}_D^\bot$&$[9,6,3]_{5}$&AMDS\\
$(5,2,1,2)$&$\mathcal{C}_D^\bot$&$[145,142,3]_{5^2}$&AMDS\\
$(5,2,2,2)$&$\mathcal{C}_D^\bot$&$[49,46,3]_{5^2}$&AMDS\\
$(13,1,1,1)$&$\mathcal{C}_D^\bot$&$[25,22,3]_{13}$&AMDS\\
\bottomrule
 \end{tabular}
\end{table}
\section{Applications in quantum codes and LCD codes}
In this section, we obtain a new class of quantum codes, which are MDS or AMDS according to the quantum Singleton bound under certain conditions. Furthermore, we obtain a new class of LCD codes, which are almost optimal according to the sphere packing bound.

\subsection{New quantum codes from self-orthogonal codes}
Let $\mathbb{C}$ be a field of complex numbers and $\mathbb{C}^q$ denote the $q$-dimensional complex vector space. Denote $(\mathbb{C}^q)^{\otimes n}=\mathbb{C}^{q^n}=\mathbb{C}^q\otimes\mathbb{C}^q\otimes\cdots\otimes\mathbb{C}^q$ as the $n$-fold tensor product of $\mathbb{C}^q$, which forms the $q^n$-dimensional complex vector space. An $((n,K,d))_q$ quantum code, denoted as $\mathcal{Q}$, is a $K$-dimensional subspace of $\mathbb{C}^{q^n}$, where $n$ is the length, $K=\dim_\mathbb{C}\mathcal{Q}$ is the dimension and $d$ is the minimum distance of $\mathcal{Q}$, respectively. Generally, the parameters of $\mathcal{Q}$ can also be expressed as $[[n,k,d]]_q$, where $k=\log_qK$.

\begin{definition}\cite[Definition 2.1.8, Page 47]{Feng2010}
An $[[n,k,d]]_q$ quantum code $\mathcal{Q}$ is said to be pure if $\langle v|e|v'\rangle=0$ for any two codewords $|v\rangle$ and $|v'\rangle$ of $\mathcal{Q}$ and any error $e\in E_n$ such that $1\leq w_\mathcal{Q}(e)\leq d-1$, where $\langle v|e|v'\rangle$ denotes the Hermitian inner product of $e|v\rangle$ and $|v'\rangle$, $E_n$ denotes the error group of $\mathcal{Q}$ and $w_\mathcal{Q}(e)$ denotes the quantum weight of $e$.
\end{definition}
By using the well-known CSS construction \cite{Calderbank1997,Steane1996} and Steane construction \cite{Steane1999}, one can construct binary quantum codes from binary self-orthogonal codes. In the following lemma, the Steane construction was generalized to the general $q$-ary case, where $q$ is a prime power.

\begin{lemma}\cite{Ling2010}\label{le4.15}
Let $\mathcal{C}_1$ and $\mathcal{C}_2$ be $[n,k_1,d_1]_q$ and $[n,k_2,d_2]_q$ linear codes, respectively. If $\mathcal{C}_1^\bot\subseteq\mathcal{C}_1\subseteq\mathcal{C}_2$ and $k_1+2\leq k_2$, then there exists an
$$\left[\left[n,k_1+k_2-n,\min\left\{d_1,\left\lceil\frac{q+1}{q}d_2\right\rceil\right\}\right]\right]_q$$
pure quantum code, where $\left\lceil\frac{q+1}{q}d_2\right\rceil$ denotes the smallest integer greater than or equal to $\frac{q+1}{q}d_2$.
\end{lemma}
\par
Similar to classical linear codes, there are trade-offs among the parameters of quantum codes as well. The well-known quantum Singleton bound on pure quantum codes is given as follows.
\begin{lemma}\cite[Quantum Singleton bound]{Feng2010}\label{le4.16}
For any $[[n,k,d]]_q$ quantum code, it parameters satisfy $2(d-1)\leq n-k$.
\end{lemma}

If an $[[n,k,d]]_q$ quantum code attains the quantum Singleton bound, i.e., $2(d-1)=n-k$, then it is called a maximum distance separable (MDS for short) quantum code. If its parameters satisfy $2d=n-k$, then it is called an almost maximum distance separable (AMDS for short) quantum code. The following lemma shows that quantum codes over $\mathbb{F}_q$ attaining the quantum Singleton bound cannot exceed a length of $q^2+1$, except in a few sporadic cases, assuming that the classical MDS conjecture holds.

\begin{lemma}\cite{Ketkar2006}\label{Quantum MDS conjecture}
If the classical MDS conjecture holds, then there are no nontrivial MDS quantum codes of lengths exceeding $q^2+1$ except when $q$ is even and $d=4$ or $d=q^2$ in which case $n\leq q^2+2$.
\end{lemma}\label{1104}

By Lemma \ref{Quantum MDS conjecture}, it is easy to find that an AMDS quantum code over $\mathbb{F}_q$ with length $n>q^2+1$ is optimal, where $q$ is odd prime power. In the following, we focus on the constructions of MDS and AMDS pure quantum codes.

\begin{theorem}\label{th4.17}
Let $s$ be a positive integer and $s_1,s_2$ be two positive divisors of $s$. If one of the following conditions holds:
\begin{itemize}
  \item [{\rm (i)}] $s_2|s_1$ and $s\geq 2s_1$;
  \item [{\rm (ii)}] $\frac{s_2}{s_1}$ is odd and $2s>s_1+s_2$;
  \item [{\rm (iii)}] $\frac{s_2}{s_1}$ is even and $2s>2s_1+s_2$,
\end{itemize}
then there exists a pure quantum code $\mathcal{Q}$ with parameters
$$\left[\left[\frac{p^{2s}+(p^{s_1}-1)\mathcal{G}^2}{p^{s_1}},\frac{p^{2s}+(p^{s_1}-1)\mathcal{G}^2}{p^{s_1}}-\frac{2s}{s_2}-2,3\right]\right]_{p^{s_2}}.$$
Furthermore, $\mathcal{Q}$ is an MDS pure quantum code if $s=s_2$ and $s_2>2s_1$, and an AMDS pure quantum code if $s=2s_2$ according to the quantum Singleton bound.
\end{theorem}

\begin{proof}
By Theorems \ref{th4.1}, it is clear that the dual code $\mathcal{C}_D^\bot$ of $\mathcal{C}_D$ is an $[n,n-(\frac{2s}{s_2}+1),3]$ linear code over $\mathbb{F}_{p^{s_2}}$, where $n=\frac{p^{2s}+(p^{s_1}-1)\mathcal{G}^2}{p^{s_1}}$. If $s,s_1$ and $s_2$ satisfy the self-orthogonal conditions given in Theorem \ref{th4.1}, then $\mathcal{C}_D\subseteq\mathcal{C}_D^\bot$. Let $\mathcal{C}_1=\mathcal{C}_D^\bot$ and $\mathcal{C}_2$ be the dual of the code $\{c\mathbf{1}:c\in\mathbb{F}_{p^{s_2}}\}\subseteq\mathcal{C}_D$, where $\mathbf{1}$ is the all-$1$ vector of length $n$. It is easy to check that $\mathcal{C}_2$ has parameters $\left[n,n-1,2\right]_{p^{s_2}}$. Then we can deduce that $\mathcal{C}_1^\bot\subseteq\mathcal{C}_1\subseteq\mathcal{C}_2$ and ${\rm dim}(\mathcal{C}_1)+2\leq{\rm dim}(\mathcal{C}_2)$. By Lemma \ref{le4.15}, then there exists a pure quantum code $\mathcal{Q}$ with parameters $[[n,n-\frac{2s}{s_2}-2,3]]_{p^{s_2}}$, where $n=\frac{p^{2s}+(p^{s_1}-1)\mathcal{G}^2}{p^{s_1}}$. Under the self-orthogonal conditions given in Theorem \ref{th4.1}, $\mathcal{Q}$ is an $[[n,n-4,3]]_{p^{s_2}}$ MDS pure quantum code if $s=s_2$ and $s_2>2s_1$ according to the quantum Singleton bound given in Lemma \ref{le4.16}. If $s=2s_2$ and $s_2|s_1$, from the self-orthogonal conditions given in Theorem \ref{th4.1}, we have $s_2\geq s_1$. This is a contradiction except that $s_2=s_1$. Hence, if $s=2s_2$, we have $s_1=s_2$, then $\mathcal{Q}$ is an $[[n,n-6,3]]_{p^{s_2}}$ AMDS pure quantum code according to the quantum Singleton bound given in Lemma \ref{le4.16}.
\end{proof}

\begin{remark}
${\rm (i)}$ Compared with the known quantum codes with minimum distance $3$ in~\cite{Bierbrauer2000,Chen2013,Guo2021,Ketkar2006,Li12023,Kai2014,Liu2010,Shi2017,Shi2018,Wang2019,Wang2024,Zhang2018,Zhu2018}, we find that the lengths of these quantum codes are different from the lengths of the quantum codes constructed in Theorem \ref{th4.17}. In addition, from Lemma \ref{Quantum MDS conjecture}, our MDS and AMDS pure quantum codes are optimal.

${\rm (ii)}$ From Theorem \ref{th4.17}, some MDS and AMDS pure quantum codes are listed in Table \ref{tab5}, which have been verified by Magma programs. Compared with the known quantum twisted codes in \cite{Bierbrauer}, we find that our codes have a higher code rate in some cases. For example, let $(p,s,s_1,s_2)=(3,2,1,1)$, then $s=2s_2,s_1|s_2$ and $2s\geq s_1+s_2$. By Theorem \ref{th4.17}, there exists a $[[33,27,3]]_3$ AMDS pure quantum code. However, the best known quantum twisted code with length $33$ and minimum distance $3$ over $\mathbb{F}_3$ has the parameters $[[33,23,3]]_3$ in \cite{Bierbrauer}.
\end{remark}

\begin{table}[h!]
\setlength{\tabcolsep}{5mm}
\caption{Some MDS and AMDS pure quantum codes from Theorem \ref{th4.17} }
\centering
\label{tab5}
 \begin{tabular}{ccccc}
\toprule
$(p,s,s_1,s_2)$&$[[n,k,d]]_{p^{s_2}}$ in Theorem \ref{th4.17}&Optimality&$[[n',k',d']]_{p^{s_2}}$ in \cite{Bierbrauer}\\
\midrule
$(3,2,1,1)$&$[[33,27,3]]_{3}$&AMDS&$[[33,23,3]]_3$\\
$(3,3,1,3)$&$[[225,221,3]]_{3^3}$&MDS&--\\
$(3,4,2,2)$&$[[801,795,3]]_{3^2}$&AMDS&$[[814,795,3]]_{3^2}$\\
$(3,4,1,4)$&$[[2241,2237,3]]_{3^4}$&MDS&--\\
$(5,2,1,1)$&$[[145,139,3]]_5$&AMDS&$[[144,134,3]]_5$\\
$(5,3,1,3)$&$[[3225,3221,3]]_{5^3}$&MDS&--\\
$(5,4,2,2)$&$[[16225,16219,3]]_{5^2}$&AMDS&--\\
$(7,2,1,1)$&$[[385,379,3]]_7$&AMDS&$[[387,375,3]]_7$\\
$(7,3,1,3)$&$[[16513,16509,3]]_{7^3}$&MDS&--\\
\bottomrule
 \end{tabular}
\end{table}
\subsection{New LCD codes from self-orthogonal codes}

It is known that a linear code $\mathcal{C}$ is called a LCD code if  $\mathcal{C}\cap\mathcal{C}^\bot=\{\mathbf{0}_n\}$, where $\mathbf{0}_n$ denotes the zero vector of length $n$. If $\mathcal{C}$ is an LCD code, then $\mathcal{C}^\bot$ is also a LCD code. In \cite{Massey1992}, Massey first introduced the concept of LCD codes and showed that a linear code is a LCD code if and only if $GG^\top$ is nonsingular, where $G$ is the generator matrix of $\mathcal{C}$. Furthermore, LCD codes have been shown to provide an optimal linear coding solution for the two-user binary adder channel \cite{Massey1992}. In \cite{Massey1998}, Massey presented fundamental properties of orthogonal, antiorthogonal, and self-orthogonal matrices, and established the relationship between these matrices and codes, which is employed to construct LCD codes. In the following, we give some definitions and necessary lemmas derived from  \cite{Massey1998}, which are useful for constructing new LCD codes based on the codes developed in this paper.
\begin{definition}
A matrix $M$ is said to be {\it row-orthogonal} if $MM^\top=I$, and {\it row-self-orthogonal} if $MM^\top=\mathbf{0}$, where $I$ and $\mathbf{0}$ denote the identity matrix and zero matrix.
\end{definition}
\par
The following lemma establishes a connection between LCD codes and row-self-orthogonal matrices.
\par
\begin{lemma}\cite{Heng2023,Massey1998}\label{le5.9}
A leading-systematic linear code $\mathcal{C}$ is an LCD code if its systematic generator matrix $G= [I : M]$ is row-orthogonal, i.e., the matrix $M$ is row-self-orthogonal.
\end{lemma}
\par
By Lemma \ref{le5.9}, it is evident that an LCD code can be derived from a self-orthogonal code.
\begin{lemma}\label{le5.10}
Let $\mathcal{C}$ be an $[n,k,d]$ linear code over $\mathbb{F}_q$ with generator matrix $G$. If $\mathcal{C}$ is self-orthogonal, then the matrix $\overline{G}=[I_k : G]$ generates an $[n+k,k,d']$ LCD code over $\mathbb{F}_q$.
\end{lemma}

In the following, a class of new LCD codes are given.
\begin{theorem}\label{th5.11}
Let $s$ be a positive integer and $s_1,s_2$ be two positive divisors of $s$. If one of the following conditions holds:
\begin{itemize}
  \item [{\rm(i)}] $s_2|s_1$ and $s\geq 2s_1$;
  \item [{\rm (ii)}] $\frac{s_2}{s_1}$ is odd and $2s>s_1+s_2$;
  \item [{\rm (iii)}] $\frac{s_2}{s_1}$ is even and $2s>2s_1+s_2$,
\end{itemize}
then the matrix $\overline{G}=[I : G_{D}]$ generates an LCD code $\mathcal{C}_{L}$ with parameters
$$\left[\frac{p^{2s}+(p^{s_1}-1)\mathcal{G}^2}{p^{s_1}}+\frac{2s}{s_2}+1,\frac{2s}{s_2}+1,d(\mathcal{C}_L)\geq d(\mathcal{C}_D)+1\right]_{p^{s_2}},$$
where $G_D$ is the generator matrix of $\mathcal{C}_D$ given by Eq.(\ref{matrix:{G_D}}) and $d(\mathcal{C}_D)$ is the minimum distance of $\mathcal{C}_D$ determined by Theorem \ref{th4.1}.
\end{theorem}

\begin{proof}
By Theorem \ref{th4.1} and Lemma~\ref{le5.10}, the linear code $\mathcal{C}_L$ generated by $\overline{G}$ is a LCD code if $s,s_1$ and $s_2$ satisfy the self-orthogonal conditions in Theorem~\ref{th4.1}. By the definition of $\mathcal{C}_L$, it is clear that the length of $\mathcal{C}_L$ is $\frac{p^{2s}+(p^{s_1}-1)\mathcal{G}^2}{p^{s_1}}+\frac{2s}{s_2}+1$, the dimension of $\mathcal{C}_L$ is $\frac{2s}{s_2}+1$ and the minimum distance $d(\mathcal{C}_L)\geq d(\mathcal{C}_D)+1$, where $d(\mathcal{C}_D)$ is the minimum distance of $\mathcal{C}_D$ determined by Theorem~\ref{th4.1}.
\end{proof}

From Theorem \ref{th5.11}, we can obtain a class of almost optimal LCD codes.
\begin{corollary}\label{coro5.11}
Let $p$ be an odd prime, $s=2$ and $s_1=s_2=1$. Let $\omega$ be a generator of $\mathbb{F}_{p^s}$. Then linear code $\mathcal{C}_L'$ with the generator matrix $\overline{G}'=[I_5:G_D']$ is a $[p^3+p^2-p+5,5, \geq p^2(p-1)+2]_p$ LCD code, where
\begin{equation}\label{matrix:{G_D}'}
\begin{aligned}
G_{D}'=
\left[
  \begin{array}{ccccc}
    1&1&\cdots &1&1\\
    {\rm Tr}_{1}^2(x_1\omega^0) & {\rm Tr}_{1}^2(x_2\omega^0) & \cdots &{\rm Tr}_{1}^2(x_{n-1}\omega^0) & 0 \\
    {\rm Tr}_{1}^2(x_1\omega)+1 & {\rm Tr}_{1}^2(x_2\omega)+1 & \cdots &{\rm Tr}_{1}^2(x_{n-1}\omega)+1 & 1 \\
    {\rm Tr}_{1}^2(y_1\omega^0) & {\rm Tr}_{1}^2(y_2\omega^0) & \cdots &{\rm Tr}_{1}^2(y_{n-1}\omega^0) & 0 \\
    {\rm Tr}_{1}^2(y_1\omega)+1 & {\rm Tr}_{1}^2(y_2\omega)+1 & \cdots &{\rm Tr}_{1}^2(y_{n-1}\omega)+1 & 1 \\
  \end{array}
\right],
\end{aligned}
\end{equation}
and $D=\{(x_1,y_1),(x_2,y_2)\ldots,(x_n,y_n)=(0,0)\}\subseteq \mathbb{F}_{p^s}^2$. Moreover, the dual code of $\mathcal{C}_L'$ is a $[p^3+p^2-p+5,p^3+p^2-p,3]_p$ almost optimal LCD code according to the sphere packing bound.
\end{corollary}

\begin{proof}
 By the definition of linear code $\mathcal{C}_D$ and Eq.(\ref{matrix:{G_D}}), the matrix $G_D'$ given by Eq.(\ref{matrix:{G_D}'}) is a generator matrix of $\mathcal{C}_D$ defined by Eq.(\ref{eq1.4}) when $s=2$ and $s_1=s_2=1$. By Theorem \ref{th4.1}, the linear code $\mathcal{C}_D$ with generator matrix $G_D'$ is a $[p^3+p^2-p,5,p^2(p-1)]_p$ self-orthogonal code. From Theorem \ref{th5.11}, we know that LCD code $\mathcal{C}_L'$ has parameters $[p^3+p^2-p+5,5,d(\mathcal{C}_L')]_p$, where $d(\mathcal{C}_L')$ is the minimum distances of $\mathcal{C}_L'$.

Let $\mathbf{c}'$ be a codeword of $\mathcal{C}_L$ with the minimum weight, then $\mathbf{c}'$ can be expressed as
\begin{equation*}
\begin{aligned}
\mathbf{c}'&=(\ell,a_0,a_1,b_0,b_1)\overline{G}'=(\ell,a_0,a_1,b_0,b_1,\mathbf{c}),
\end{aligned}
\end{equation*}
where $\ell,a_0,a_1,b_0,b_1\in\mathbb{F}_{p}$, $a=a_0+a_1\omega$, $b=b_0+b_1\omega$ and $\mathbf{c}=\left({\rm Tr}_{1}^2(ax+by)+\ell+a_1+b_1\right)_{(x,y)\in D}$ is a codeword of $\mathcal{C}_D$ with minimum weight.
By Theorem \ref{th4.1}, we know that $wt(\mathbf{c})=d(\mathcal{C}_D)=p^2(p-1)$ if and only if $a,b,\ell+a_1+b_1$ are all in the set $$\{a,b\in\mathbb{F}_{p^2},\ell+a_1+b_1\in\mathbb{F}_{p}:(a,b)\neq(0,0),\,\Gamma(s_1)=0~{\rm and}~\ell+a_1+b_1=0\}.$$
Then we can deduce that the minimum weight $wt(\ell,a_0,a_1,b_0,b_1)$ equals to $2$ when $a_1=b_1=\ell=0$ and $a_0\neq0,b_0\neq0$. Hence, $wt(\ell,a_0,a_1,b_0,b_1)\geq2$. This implies that $$d(\mathcal{C}_L')\geq d(\mathcal{C}_D)+2=p^2(p-1)+2.$$ Therefore, the linear code $\mathcal{C}_L'$ is a $[p^3+p^2-p+5,5,\geq p^2(p-1)+2]_p$ LCD code.

In the following, we determine the parameters of the dual code of $\mathcal{C}_L'$.  By the definition of $\mathcal{C}_L'$,  it is obvious that the length of $\mathcal{C}_L'^\bot$ is $p^3+p^2-p+5$, and the dimension of $\mathcal{C}_L'^\bot$ is $p^3+p^2-p$, and $d(\mathcal{C}_L'^\bot)\leq d(\mathcal{C}_D^{\perp})=3$. Moreover, it is clear to deduce that any two columns in the first $5$ columns of $\overline{G}'$ are $\mathbb{F}_{p}$-linearly independent since $I_5$ is the identity matrix of order $5$, and any column in the first $5$ columns and the last column of $\overline{G}'$ are $\mathbb{F}_{p}$-linearly independent. Hence, in order to prove that the minimum distance of $\mathcal{C}_L'^\bot$ is $3$, we only need to show that any column $\mathbf{m}_1$ in the first $5$ columns of matrix $\overline{G}'$ and any column $\mathbf{m}_2$ in the first $(p^3+p^2-p-1)$ columns of $G_{D}'$ are $\mathbb{F}_{p}$-linearly independent.

If $\mathbf{m}_1\neq(1,0,\ldots,0)^\top$, it is easy to check that $\mathbf{m}_1$ and $\mathbf{m}_2$ are $\mathbb{F}_{p}$-linearly independent. If $\mathbf{m}_1=(1,0,\ldots,0)^\top$, then $\mathbf{m}_1$ and $\mathbf{m}_2$ are $\mathbb{F}_{p}$-linearly dependent if and only if
\begin{equation}\label{System}
\begin{aligned}
\left\{
  \begin{array}{ll}
    x_i+x_i^p=0, & \hbox{} \\
    \omega x_i+\omega^px_i^p+1=0, & \hbox{} \\
    y_i+y_i^p=0, & \hbox{} \\
    \omega y_i+\omega^py_i^p+1=0, & \hbox{} \\
  \end{array}
\right.
\end{aligned}
\end{equation}
where $(x_i,y_i)\in D\backslash\{(0,0)\}$ for some $1\leq i\leq |D|-1$ and $D$ is given in (\ref{eq1.5}). It is easy to check that $x_i=y_i$ if Eq.(\ref{System}) holds for some $1\leq i\leq |D|-1$. Since $(x_i,y_i)\in D$, we can deduce that
$${\rm Tr}_{1}^2(x_i^2+y_i^2)=2{\rm Tr}_{1}^2(x_i^2)=2(x_i+x_i^p)^2-4x_i^{p+1}=-4x_i^{p+1}\neq0,$$
which contradicts the fact that $(x_i,y_i)\neq (0,0)$. Hence, the minimum distance of the dual code $\mathcal{C}_L$ is $3$.

Obviously, we have
$$p^{n-(p^3+p^2-p+2)}<1+n(p-1),$$
where $n=p^3+p^2-p+5$.
Then from the sphere packing bound given in Lemma \ref{The sphere packing bound}, we know that
 there does not exist a  $[p^3+p^2-p+5,p^3+p^2-p+2,3]_p$ code. Hence, the dual code $\mathcal{C}_L'^\bot$ is a $[p^3+p^2-p+5,p^3+p^2-p,3]_p$ almost optimal LCD code.
\end{proof}

Below, we present an example that support the results of Theorem \ref{th5.11} and Corollary \ref{coro5.11}, which are verified by Magma programs.
\begin{example}
Let $(p,s,s_1,s_2)=(3,2,1,1)$. Then $s_2|s_1$ and $s\geq2s_1$ (or $\frac{s_2}{s_1}$ is odd, $2s\geq s_1+s_2$). By Theorem \ref{th4.1} and Magma programs, $\mathcal{C}_D$ is a $[33,5,18]_3$ self-orthogonal code with generator matrix
\begin{align*}
\setlength{\arraycolsep}{1.2pt}
\begin{array}{lc}
G_D=\left[
      \begin{array}{c}
        \mathbf{r}_1 \\
        \mathbf{r}_2 \\
        \mathbf{r}_3 \\
        \mathbf{r}_4 \\
        \mathbf{r}_5 \\
      \end{array}
    \right]
=
     \begin{bmatrix}
  \begin{array}{ccccccccccccccccccccccccccccccccc}
      1& 1& 1& 1& 1& 1& 1& 1& 1& 1& 1& 1& 1& 1& 1& 1& 1& 1& 1& 1& 1& 1& 1& 1& 1& 1& 1& 1& 1& 1& 1& 1& 1\\
      2& 2& 1& 1& 1& 1& 1& 0& 0& 1& 1& 1& 1& 1& 1& 1& 2& 2& 2& 2& 2& 0& 0& 2& 2& 2& 2& 2& 0& 0& 0& 0& 0\\
      1& 1& 0& 0& 0& 0& 0& 1& 1& 1& 1& 1& 1& 1& 2& 2& 0& 0& 0& 0& 0& 2& 2& 2& 2& 2& 2& 2& 0& 0& 0& 0& 0\\
      0& 0& 1& 1& 2& 2& 0& 2& 1& 1& 1& 2& 2& 0& 0& 0& 1& 1& 2& 2& 0& 2& 1& 1& 1& 2& 2& 0& 1& 1& 2& 2& 0\\
      1& 2& 0& 1& 0& 2& 0& 1& 2& 0& 1& 0& 2& 0& 1& 2& 0& 1& 0& 2& 0& 1& 2& 0& 1& 0& 2& 0& 0& 1& 0& 2& 0\\
  \end{array}
\end{bmatrix}
\end{array}.
\end{align*}
By Theorem \ref{th5.11}, the matrix $\overline{G}=[I_5:G_D]$ generates a LCD code $\mathcal{C}_L$ with parameters $[38,5,19]_3$. By Corollary \ref{coro5.11}, we can get that the linear code $\mathcal{C}_D'$ with the generator matrix $G_D'=[\mathbf{r}_1,\mathbf{r}_2,\mathbf{r}_3+\mathbf{r}_1,\mathbf{r}_4,\mathbf{r}_5+\mathbf{r}_1]$ is a $[33,5,18]_3$ self-orthogonal code. Then the matrix $\overline{G}'=[I_5: G_D']$ generates a LCD code $\mathcal{C}_L'$ with parameters $[38,5,20]_3$. Moreover, the dual code of $\mathcal{C}_L'$ is a $[38,33,3]_3$ almost optimal LCD code according to the sphere packing bound given in Lemma \ref{The sphere packing bound}.
\end{example}

\begin{remark}
Compared with known infinite
families of LCD codes constructed by \cite{Heng2023,Huang2023,Li12017,Li2018,Li2019,Heng12024,Mesnager2021,Wang2024,Wu2021,Yan2018}, we have not identified any class of LCD codes with the same parameters as those in Theorem \ref{th5.11} and Corollary \ref{coro5.11}.
\end{remark}

\section{Concluding remarks}
In this paper, we investigated a class of self-orthogonal linear codes using the defining-set approach. The main results are summarized as follows.
\begin{itemize}
\item By utilizing the Gaussian sums over finite fields, we determined the weight distribution of the linear code $\mathcal{C}_D$ defined by Eq.(\ref{eq1.4}), and showed that $\mathcal{C}_D$ is self-orthogonal under some
  certain conditions. Furthermore, the parameters of dual code $\mathcal{C}_D^\bot$ were determined and a class of AMDS codes from their duals are obtained. (See Theorem \ref{th4.1}).
\item By utilizing the self-orthogonal property of $\mathcal{C}_D$, we obtained a class of new quantum codes, which are MDS or AMDS according to the quantum Singleton bound under certain conditions. Furthermore, we obtained a new class of LCD codes, which are almost optimal according to the sphere packing bound. (See Theorems \ref{th4.17}, \ref{th5.11} and Corollary \ref{coro5.11}).
\item According to the tables of best codes known in \cite{Grassl}, some optimal and AMDS linear codes are listed in Table \ref{tab4}. Moreover, some MDS and AMDS pure quantum codes are listed in Table \ref{tab5}.
\end{itemize}

Comparing our constructed linear code with the codes in \cite{Ding2020,Heng2017,Heng2020,Heng2023,Heng2024,Li2016,Li12023,Heng12024,Liu2019,Sun20241,Zheng2015,Zheng2021} and the references therein,
we have not yet found linear codes with the same parameters and weight distribution as in this paper. In this sense, the linear codes constructed in this paper are new.

\section*{References}
\bibliographystyle{model1a-num-names}

\end{document}